\documentclass[11pt]{article}
\usepackage[letterpaper,hmargin=1in,vmargin=1in]{geometry}
\usepackage{graphicx}
\usepackage{booktabs}
\usepackage{adjustbox,multirow}
\usepackage{epstopdf}
\usepackage{url}
\usepackage{hyperref}
\usepackage{amsmath,amsfonts,latexsym,verbatim,threeparttable,amssymb,amsfonts,dsfont,rotating,pifont}
\usepackage{array,hhline}
\usepackage{times}
\usepackage{verbatim}
\usepackage{wrapfig}
\usepackage{caption}
\usepackage{subcaption}
\usepackage{multicol}

\usepackage{tikz}
\usetikzlibrary{calc} 
\usetikzlibrary{arrows}
\usepackage{verbatim}

\usepackage{pgfplots}
\usepackage{url,times,verbatim,multirow}
\usepackage{amsfonts,amsmath,amssymb, amsthm,color,wasysym}
\usepackage{latexsym,verbatim,threeparttable,rotating} 
\usepackage{hyperref}
\usepackage{array,hhline} 	
\usepackage{amsmath,footnote}
\usepackage{xcolor}
 \usepackage{etex}
\hypersetup{
     colorlinks   = true,
     citecolor    = gray
}
\date{}

\newtheorem{theorem}{Theorem}[section]
\newtheorem{corollary}[theorem]{Corollary}
\newtheorem{lemma}[theorem]{Lemma}

\newtheorem{definition}[theorem]{Definition}

\newtheorem{claim}[theorem]{Claim}

\usepackage{color}
\usepackage{framed}

  \newcommand{\commentAS}[1]{\textcolor{red}{{\sf (Ajith's Note:)} {\sl{#1}}}}
  \newcommand{\commentP}[1]{\textcolor{magenta}{{\sf (Pratik's Note:)} {\sl{#1}}}}

\mathchardef\mhyphen="2D

\newcommand{\subscript}[1]{{{#1}}}
\newcommand{\vset}[3]{(#1_{#2},\cdots ,#1_{#3})}

\newcommand{\xor}{\oplus}
\newcommand{\Xor}{\bigoplus}
\newcommand{\band}{\odot}

\newcommand{\Order}{\mathcal{O}}

\newcommand{\define}{\ensuremath{:=}}

\def\cross{\times}
\newcommand{\from}{\leftarrow}

\newcommand{\negl}{\mathtt{negl}}
\newcommand{\hdi}{\mathtt{HDI}}
\newcommand{\puncture}{\mathtt{PRN}}

\newcommand{\ssec}{\mu}
\newcommand{\csec}{\kappa}

\newcommand{\A}{\mathcal{A}}
\newcommand{\B}{\mathcal{B}}

\newcommand{\I}{\mathcal{I}}

\newcommand{\ot}[3]{{{#3} \choose 1}\textsf{-OT}_{#1}^{#2}}
\newcommand{\sen}{\textsf{S}}
\newcommand{\rec}{\textsf{R}}

\newcommand{\WH}{\texttt{WH}}

\newcommand{\accept}{\texttt{Accept}}
\newcommand{\reject}{\texttt{Reject}}
\newcommand{\abort}{\texttt{Abort}}

\newcommand{\prob}{\ensuremath{\mathsf{Pr}}}
\newcommand{\poly}{\ensuremath{\mathsf{Poly}}}
\newcommand{\Fault}{\ensuremath{\mathcal{T}}}

\newcommand{\PassCheck}{\ensuremath{\mathsf{PC}}}
\newcommand{\Dist}{\ensuremath{\mathsf{D}}}
\newcommand{\DecodeFail}{\ensuremath{\mathsf{FE}}}

\newcommand{\vzero}{\ensuremath{\mathbf{0}}}
\newcommand{\advA}[2]{\ensuremath{Adv_\subscript{#1}[#2]}}

\newcommand{\mA}{\ensuremath{\mathbf{A}}}
\newcommand{\mB}{\ensuremath{\mathbf{B}}}

\newcommand{\mD}{\ensuremath{\mathbf{D}}}
\newcommand{\mE}{\ensuremath{\mathbf{E}}}

\newcommand{\ve}{\mathbf{e}}
\newcommand{\va}{\mathbf{a}}
\newcommand{\vb}{\mathbf{b}}
\newcommand{\vc}{\mathbf{c}}
\newcommand{\vd}{\mathbf{d}}
\newcommand{\vk}{\mathbf{k}}
\newcommand{\vp}{\mathbf{p}}

\newcommand{\vs}{\mathbf{s}}

\newcommand{\vx}{\mathbf{x}}
\newcommand{\vy}{\mathbf{y}}
\newcommand{\vz}{\mathbf{z}}
\newcommand{\vw}{\mathbf{w}}

\newcommand{\C}{\mathcal{C}}
\newcommand{\lf}{\mathcal{L}}

\newcommand{\F}{\mathbb{F}}

\newcommand{\N}{\mathbb N}

\newcommand{\X}{\mathcal{X}}

\newcommand{\Func}{\mathcal{F}}
\newcommand{\Sim}{\mathcal{S}}
\newcommand{\Env}{\mathcal{Z}}
\newcommand{\Adv}{\mathcal{A}}
\newcommand{\Ideal}{\ensuremath{\textsc{IDEAL}}}
\newcommand{\Real}{\ensuremath{\textsc{REAL}}}
\newcommand{\Hyb}{\ensuremath{\textbf{\textsc{HYB}}}}

\newcommand{\FOT}[3]{\ensuremath{\mathcal{F}^{({#3},{#2},{#1})}_{\textsf{OT}}}}

\newcommand{\VIEW}[7]{\ensuremath{\mathcal{V}^{\textsc{#1}}_{#2, #3, #4}(#5, #6, #7)}}

\newcommand{\FCOIN}{\ensuremath{\mathcal{F}_\textsc{COIN}}}
\newcommand{\FPSI}{\ensuremath{\mathcal{F}_\textsf{PSI}}}
\newcommand{\coin}{\texttt{coin}}
\newcommand{\FRAND}{\ensuremath{\mathcal{F}_\textsc{RAND}}}

\newcommand{\XANDY}{\ensuremath{\vx \cap \vy}}
\newcommand{\PIS}{\ensuremath{\Pi}}

\newcommand{\figlab}[1]{\label{fig:#1}}

\newenvironment{boxfig}[2]{
	\begin{figure}[htb!]		
		\fontsize{5}{5}\selectfont
		\newcommand{\FigCaption}{#1}
		\newcommand{\FigLabel}{#2}
		\vspace{-.10cm}
		\begin{center}
			\caption{\FigCaption}
			\begin{small}			 
				\begin{adjustbox}{max width=\textwidth}
					\begin{tabular}{@{}|@{~~}l@{~~}|@{}}
						\hline
						\rule[-1ex]{0pt}{1ex}\begin{minipage}[b]{.95\linewidth}
							\vspace{1ex}	
						}{%
					\end{minipage}\\
					\hline
				\end{tabular}	
			\end{adjustbox}		
		\end{small}
		\vspace{-0.25cm}
		\figlab{\FigLabel}
	\end{center}
	\vspace{-.38cm}
\end{figure}
}

\newenvironment{boxfig*}[2]{
	\begin{figure*}[h!]		
		\fontsize{5}{5}\selectfont
		\newcommand{\FigCaption}{#1}
		\newcommand{\FigLabel}{#2}
		\vspace{-.10cm}
		\begin{center}
			\caption{\FigCaption}
			\begin{small}			 
				\begin{adjustbox}{max width=\textwidth}
					\begin{tabular}{@{}|@{~~}l@{~~}|@{}}
						\hline
						\rule[-1ex]{0pt}{1ex}\begin{minipage}[b]{.95\linewidth}
							\vspace{1ex}	
						}{%
					\end{minipage}\\
					\hline
				\end{tabular}	
			\end{adjustbox}		
		\end{small}
		\vspace{-0.25cm}
		\figlab{\FigLabel}
	\end{center}
	\vspace{-.38cm}
\end{figure*}
}

\renewenvironment{proof}{\noindent\textit{Proof.} }{\qed}

\newcommand{\bitset}{\{0,1\}}

\mathchardef\mhyphen="2D

\newcounter{itemcount}

\newenvironment{mydescription}
{\setcounter{itemcount}{0}\begin{list}
{\arabic{itemcount}.}{\usecounter{itemcount} \itemindent=-0.5cm
\itemsep=0.0in
\parsep=0.0in
\topsep=5pt
\partopsep=0.0in}}{\end{list}}

\newenvironment{mydescnoindent}
{\setcounter{itemcount}{0}\begin{list}
{\arabic{itemcount}.}{\usecounter{itemcount} 
\itemsep=0.05in
\parsep=0.0in
\topsep=3pt
\partopsep=0.0in}}{\end{list}}

\newcommand{\bool}{\{0,1\}}
\newcommand{\ppt}{\textsf{PPT}}
\newcommand*{\ovA}[1]{\overline{#1\raisebox{2mm}{}}}

\begin{document}
\title{\bf Fast Actively Secure OT Extension for Short Secrets}

  \author{Arpita Patra \thanks{Indian Institute of Science. Email: \tt{arpita@csa.iisc.ernet.in}.}  \and Pratik Sarkar \thanks{Indian Institute of Science. Email: \tt{pratik.sarkar@csa.iisc.ernet.in}.}  \and  Ajith Suresh \thanks{Indian Institute of Science. Email: \tt{ajith.s@csa.iisc.ernet.in}.}}


\maketitle

\begin{abstract}
Oblivious Transfer (OT) is one of the most fundamental cryptographic primitives  with wide-spread application in general secure multi-party computation (MPC) as well as in a number of tailored and special-purpose problems of interest such as private set intersection (PSI), private information retrieval (PIR), contract signing to name a few.  Often the instantiations of OT require prohibitive  communication and computation complexity. OT extension protocols are introduced to compute a very large number of OTs  referred as {\em extended} OTs at the cost of  a small number of OTs referred as {\em seed} OTs. 

We present a fast OT extension protocol for small secrets in active setting. Our protocol when used to produce $1$-out-of-$n$ OTs  outperforms all the known actively secure OT extensions. Our protocol is built on the semi-honest secure extension protocol of Kolesnikov and Kumaresan of CRYPTO'13 (referred as KK13 protocol henceforth) which is the best known OT extension for  short secrets. At the heart of our protocol lies an efficient consistency checking mechanism that relies on the linearity of Walsh-Hadamard (WH) codes. Asymptotically, our protocol adds a communication overhead of $\Order(\ssec \log{\csec})$ bits over KK13 protocol irrespective of the number of  extended OTs, where $\csec$ and $\ssec$ refer to computational and statistical security parameter respectively. 
 Concretely, our protocol when used to generate a large enough number of OTs adds only $0.011\mbox{-}0.028\%$  communication overhead and  $4\mbox{-}6\%$ runtime overhead both in LAN and WAN over KK13 extension. The runtime overheads drop below $2\%$ when  in addition the number of inputs of the sender in the extended OTs is large enough.

As an application of our proposed extension protocol, we show that it can be used to obtain the most efficient PSI protocol secure against a malicious receiver and a semi-honest sender. 
\end{abstract}

\section{Introduction}
Oblivious Transfer (OT)~\cite{NaorP05,Kilian88,BrassardCR86,EvenGL85,Rabin81} is perhaps the most fundamental primitive in cryptographic protocol theory. It is a two party protocol between a sender $\sen$ and a receiver $\rec$. The sender holds an array of inputs and the receiver holds an index indicating its intended pick from the sender's array. OT allows the sender to send the receiver's selected input while preserving the secrecy of the sender's other inputs on the one hand and the choice of the receiver on the other.   The  necessity and sufficiency of OT for secure multi-party computation (MPC)~\cite{Kilian88,GoldreichV87,GoldreichMW87,Yao86} backs the theoretical importance of OT. On the practical front, OT has been pivotal in  building  several state-of-the-art practically efficient general MPC protocols \cite{Lindell16,LindellP15,LindellR15,HuangKKKM14,FrederiksenJNNO13,ShelatS13,NielsenO09}  and several protocols for  special-purpose problems of interest such as private set intersection (PSI) ~\cite{PinkasSZ15,PinkasSZ14,DongCW13}. 
There is a fundamental limitation to OT's efficiency as it is unlikely that OT is possible without public-key cryptography and  solely relying on symmetric-key cryptography~\cite{ImpagliazzoR89}. The OT extension protocols \cite{KellerOS15, AsharovL0Z15, KolesnikovK13, AsharovL0Z13, NielsenNOB12, IshaiKNP03, Beaver96} have been introduced to theoretically circumvent the above limitation of OTs. They produce a large number of  OTs  referred as {\em extended} OTs from a small number of OTs referred as {\em seed} OTs and symmetric-key primitives.  When the goal is to generate a large number of OTs which is usually the case for the applications of OT, the amortized cost of generating a single OT via OT extensions turns out to be a constant number of symmetric-key operations. So most of the known practically efficient general and special-purpose MPC protocols are byproduct of concretely efficient OT extension protocols.   Of particular interest to cryptographic community are the following variants of OT: {\sf (a)} In a $1$-out-of-$2$ OT \cite{EvenGL85}, $\sen$ holds two inputs $x_0,x_1$, and $\rec$ holds a {\em choice bit} $b$. The output to $\rec$ is $x_b$ and no other party learns anything. {\sf (b)} A straight-forward extension of  $1$-out-of-$2$ OT is $1$-out-of-$n$ OT \cite{BrassardCR86} where $\sen$ holds $n$ inputs and $\rec$ holds a choice index of $\log{n}$ bits. While the first kind finds application in MPC \cite{GoldreichMW87,Yao82b}, the second kind is useful in PSI \cite{PinkasSZ15,PinkasSZ14}, symmetric PIR \cite{NaorP05}, and oblivious sampling \cite{NaorP05}, oblivious polynomial evaluation \cite{NaorP99}. As discussed below, attempts have been made to construct OT extension  protocols to output both the above kinds of OTs.   


\subsection{OT Extensions}
The theoretical feasibility of OT extension was proved by Beaver \cite{Beaver96}.   Ishai, Kilian, Nissim and Petrank \cite{IshaiKNP03} (referred as IKNP protocol henceforth) presented the first efficient OT extension protocol that builds on $\csec$ seed OTs and requires computing and sending just two hash values per extended OT.  In \cite{AsharovL0Z13},  IKNP protocol has seen several optimizations  that boost both its communication and computation complexity. Specifically, the communication per extended OT is brought down to one hash value for a special case where the extended OTs are needed for random inputs of the sender.  The computation bottleneck for implementing matrix transposition is tackled by introducing a cache-oblivious algorithm. Yet another contribution from \cite{AsharovL0Z13} is their crucial observation that the actual bottleneck in the runtime of IKNP protocol results from its communication time, particularly in wide area networks (WANs) that have high latency and low bandwidth.   In a first of its kind approach, Kolesnikov and Kumaresan~\cite{KolesnikovK13} (referred as KK13 protocol henceforth) presented an OT extension protocol that  outputs $1$-out-of-$n$ OTs starting from $2\csec$ $1$-out-of-$2$ seed OTs  and relying on specifics of Walsh-Hadamard (WH) codes.   KK13 protocol improves over all its predecessors (including IKNP) customized to generate $1$-out-of-$n$ OTs   by a factor $\Order(\log{n})$ in communication complexity when the inputs of the extended OTs are of short size.  So far KK13 protocol remains to be the most efficient way of generating $1$-out-of-$n$ OTs  for short inputs.   All the protocols discussed above work when the adversary is assumed to be {\em semi-honest}. A passive or semi-honest adversary  follows the protocol specification but attempts to learn more than allowed by inspecting the protocol transcript. An adversary is referred as  {\em active} or {\em malicious} when it behaves in any arbitrary way in an attempt to break the security of the protocol. 

OT extension literature finds numerous attempts to achieve active security. All the attempts  restrict their attention in transforming the semi-honest secure  IKNP protocol to an actively secure one. Since the IKNP protocol is resilient to any malicious behavior of the sender, an actively secure IKNP style protocol needs to enforce honest behaviour for the receiver.  Adding consistency checks for the receiver has been the strategy followed in all  the known constructions.  The efficiency (both communication and computation wise) of the consistency checks defines the overhead for an actively secure IKNP style protocol.  The consistency check introduced in \cite{IshaiKNP03}  employs expensive cut-and-choose technique on $\ssec$ parallel instances of the semi-honest IKNP protocol  where $\ssec$ is a statistical security parameter. \cite{HarnikIKN08, Nielsen07} proposes consistency checks per extended OTs.  This is improved in \cite{NielsenNOB12}  where the checks are done per seed OT. In order to tackle information leakage in their consistency check, \cite{NielsenNOB12}  needs to start with $\frac{8}{3} \csec$ seed OTs which is $\frac{8}{3}$ times more than what IKNP protocol needs. This inflates their concrete communication and computation complexity by the same factor.  \cite{AsharovL0Z15} improves  over \cite{NielsenNOB12} by trading  computation in consistency checks for a reduced number of seed OTs. Namely, the OT extension of \cite{AsharovL0Z15}  requires $\csec + 1.55\ssec$ seed OTs compared to $\frac{8}{3}\csec$  of   \cite{NielsenNOB12} and thus improves the communication done via seed OTs. In a recent work, \cite{KellerOS15} reports the most efficient actively secure IKNP style protocol that brings back the number of seed OTs to $\csec$ and handles the information leakage in the consistency check by sacrificing $\csec + \ssec$ extended OTs. The check requires an $\Order(\csec)$ bits communication irrespective of the number of extended OTs and two finite field operations per extended OT. 

Above we concentrated on practically efficient OT extension literature. Some interesting theoretical questions on OT extension are addressed in~\cite{Larraia14,LindellZ13}.

\subsection{Our Contribution}
We present an actively secure OT extension for short secrets building upon the semi-honest secure protocol of \cite{KolesnikovK13}. Like KK13 protocol, our extension protocol turns $1$-out-of-$2$ seed OTs to $1$-out-of-$n$ extended OTs. Similar to IKNP protocol, KK13 protocol is  secure against any malicious behaviour of sender but falls apart in the face of a maliciously corrupt receiver. We present a concrete attack on KK13 and add an efficient  consistency check to enforce correct behaviour of the receiver. Our check relies on the linearity of WH codes. Combined with an additional trick, our efficient consistency check incurs a communication of $\Order(\ssec \log{\csec})$ bits irrespective of the number of generated extended OTs.  Asymptotically, our OT extension matches the KK13 protocol in every respect.  Table~\ref{tab:compareallasym} shows the efficiency of various OT extension protocols  achieving $2^{-\csec}$ computational security and $2^{-\ssec}$ statistical security for producing  $m$ $1$-out-of-$n$ OTs with $\ell$-bit inputs of the sender. The following parameters have been used for comparison: {\sf (i)} number of seed OTs, {\sf (ii)} communication complexity and  {\sf (iii)} computation complexity in terms of number of hash value computations. 

\begin{table}[htb!]
	\centering
	\caption{\footnotesize Asymptotic cost of various OT extensions for producing  $m$ $1$-out-of-$n$ OTs with $\ell$-bit inputs of the sender and for achieving $2^{-\csec}$ computational security and $2^{-\ssec}$ statistical security.}
	\label{tab:compareallasym}
	\resizebox{.48\textwidth}{!}{
		\begin{tabular}{  l | l | l | l } 
			\toprule
			Reference				& \# Seed OTs  & Communication (bits) / & Security \\ 
			&    &  Computation (\# hashes)&   \\
			\midrule
			~\cite{KolesnikovK13}	& $2\csec$ &     $\Order(m(\csec + n \ell)) $  &	semi-honest \\  
			~\cite{IshaiKNP03}   	&  $\csec$ &     $\Order(m(\csec \log n + n \ell)) $ &   semi-honest \\ 
			~\cite{NielsenNOB12}	&  $\frac{8}{3}\csec$ &  $\Order(m(\csec \log n + n \ell)) $   &  active   \\ 
			~\cite{AsharovL0Z15} 	& $\csec + 1.55\ssec$ &    $\Order(m(\csec \log n + n \ell)) $  & active  \\ 
			~\cite{KellerOS15}		& $\csec$ &    $\Order(m(\csec \log n + n \ell)) $ &  active   \\ 
			This Paper			& $2\csec$ &   $\Order(m(\csec + n \ell)) $ &   active \\  
			\bottomrule
		\end{tabular}
	}	
\end{table}

Concretely, our protocol when used to generate large enough number of OTs adds only $0.011\mbox{-}0.028\%$  communication overhead and  $4\mbox{-}6\%$ runtime overhead both in LAN and WAN over KK13 protocol. The runtime overheads drop below $2\%$ when  in addition the number of inputs of the sender in the extended OTs is large enough. Our construction put in the context of other OT extensions are presented in Table~\ref{tab:commall}. The table presents figures for generating $1.25 \times 10^6$ $1$-out-of-$16$ OTs with sender's input length as $4$ bits.  The overheads are calculated with respect to KK13 protocol.  The implementation of \cite{KellerOS15} is not available in the same platform as the other OT extensions given in the table. As  per the claim made in \cite{KellerOS15}, the runtime of their OT extension bears an overhead of $5\%$ with respect to IKNP protocol both in LAN and WAN. So the runtime and overhead  in runtime of \cite{KellerOS15} with respect to KK13 protocol are calculated based on that claim.  As evident from Table~\ref{tab:commall},  our protocol when used to compute $1$-out-of-$n$ OTs  with short inputs of the sender outperforms all the known actively secure OT extensions and secures the second best spot among all the OT extension protocols listed in Table~\ref{tab:commall}  closely trailing KK13 which is the overall winner. More elaborate empirical results supporting the above claim with varied number of extended OTs and varied number of inputs of the sender in the extended OTs appear later in the paper.

\begin{table*}[h!]			
	\centering
	\caption{\footnotesize Concrete cost of various OT extension protocols for producing $1.25 \times 10^6$ $1$-out-of-$16$ OTs with sender's input length as $4$ and for achieving computational security of $2^{-128}$ and  statistical security of $2^{-40}$.}
	\label{tab:commall}
	\resizebox{.75\textwidth}{!}{
		\begin{tabular}{ l | l | l | l | l | l | l | l } 
			\toprule
			\multirow{2}{*}{Reference} & \multirow{2}{*}{\# Seed}  & Comm  & \multicolumn{2}{| c}{Runtime (in sec)} & \multicolumn{3}{| c }{Overhead w.r.t.  \cite{KolesnikovK13} (in \%)} \\
			\hhline{~~~-----}
			&   OTs & (in MB)  & LAN & WAN & Communication & Runtime in LAN & Runtime in WAN \rule{0pt}{3ex}  \\  
			\midrule
			~\cite{KolesnikovK13}	& $256$	& $47.69$	 	& $21.68$	 & 	$115.34$   & $0$  & $0$  	& $0$ \\ 
			~\cite{IshaiKNP03}   	& $128$ 	& $87.74$ 	& $24.07$	  & $133.81$	 & $84$  & $11.02$  & $16$ \\
			~\cite{NielsenNOB12}	& $342$	& $215.95$ 	&  $24.84$  & $143.20$    & $352.7$ & $14.6$ & $24.14$ \\
			~\cite{AsharovL0Z15} 	& $190$	& $166.54$ 	&  $24.81$ & $158.6$	& $249.1$ & $14.4$  & $37.5$ \\
			~\cite{KellerOS15} 	& $128$	&  $> 87.74$	&   $>25.27$	 & $>140.5$	& $> 84$ & $>16.5$  & $>21.8$ \\
			This Paper 			& $256$ 	& $47.70$		& $22.50$   &  $121.94$   & ${\bf 0.028}$ &  ${\bf 3.78}$	 & ${\bf 5.72}$ \\ [1ex] 
			\bottomrule
		\end{tabular}
	}
	
\end{table*}

Lastly, the OT extensions presented in all the works in the table except  \cite{KolesnikovK13} inherently produce $1$-out-of-$2$ OTs. The transformation  from  $1$-out-of-$2$ to $1$-out-of-$n$ OT given in \cite{NaorP05} is used to transform their extended OTs to $1$-out-of-$n$ OTs.  The transformation that works for reverse direction  \cite{NaorP05} is unfortunately {\em not} maliciously secure. This prevents us from claiming a similar gain  when our protocol is used to generate $1$-out-of-$2$ OTs. We leave open the question of finding an efficient actively secure transformation from  $1$-out-of-$n$ to $1$-out-of-$2$ OT.

We show an interesting application of our proposed extension protocol in OT-based PSI protocols. Specifically, we use our maliciously secure OT extension in the PSI protocol of  \cite{cryptoeprint:2016:665} to obtain the most efficient PSI protocol  that is maliciously secure against a corrupt receiver and semi-honestly secure against a corrupt  sender.  In brief, a PSI protocol between two parties, namely a sender $\sen$ and a receiver $\rec$ holding sets $X = \{x_1, x_2, \ldots x_{n_1}\}$ and $Y= \{y_1, y_2,  \ldots y_{n_2} \}$ respectively, outputs  the intersection $X \cap Y$ to the receiver and nothing to the sender.  As evident from the theoretical and experimental results presented in this work, our maliciously secure OT extension protocol is a better choice compared to the existing maliciously secure extension protocols ~\cite{AsharovL0Z15,NielsenNOB12,KellerOS15}   when $1$-out-of-$n$ OTs  are required as output. As PSI employs $1$-out-of-$n$ (instead of $1$-out-of-$2$) OTs,  our extension protocol fits the bill. Lastly, we find a concrete vulnerability for the malicious corrupt receiver case in  Lamb{\ae}k's PSI protocol when semi-honest KK13 OT protocol is used in it. This confirms Lamb{\ae}k's concern of privacy breach of his PSI protocol that may result from privacy breach of the underlying OT protocols and further confirms the necessity of  maliciously secure OT extension in Lamb{\ae}k's PSI protocol. 

\section{Preliminaries}
\label{sec:prelim}
We present below the required preliminaries and techniques. We revisit the KK13 \cite{KolesnikovK13} protocol  and present a concrete attack on it in Section~\ref{sec:kk13}. Next, our proposed actively secure protocol with efficiency and security analysis is presented in Section~\ref{malkk13}. Section~\ref{sec:imple} shows our empirical findings and analysis.  Lastly,  the application of our actively secure protocol in PSI appear in Section \ref{sec:psi}.

\subsection{Notations}
We use $\xor$ to denote bitwise XOR operation and $\band$ to denote bitwise AND operation. We denote vectors in bold smalls and matrices in bold capitals. For a matrix $\mA$, we let $\va_j$ denote the $j$th row of $\mA$, and $\va^i$ denote the $i$th column of $\mA$. For a vector $\va$, $a_i$ denotes the $i$th element in the vector. For two vectors $\va$ and $\vb$ of length $p$,  we use the notation $\va \xor \vb$ to denote the vector $(a_1 \xor b_1,\cdots, a_p \xor b_p)$   and the notation $\va \band \vb$ to denote the  vector $(a_1 \band b_1, \cdots, a_p \band b_p)$. The notation $\Xor\limits_{j=1}^m \va_j$ denotes the XOR of $m$ vectors, i.e. $\va_1 \xor \cdots \xor \va_m$. We denote by $\va \otimes \vb$ the inner-product value $\Xor\limits_{i=1}^p a_i \band b_i $. Finally, suppose $c \in \bool$, then $c \band \va$ denotes the vector $(c \band a_1,\cdots, c \band a_p)$.  We denote by $a \gets_R A$ the random sampling of $a$ from a distribution $A$. We denote by $[x]$, the set of elements $\{1,\ldots,x\}$. 

We denote by $\hdi$ a function that takes two binary vectors of same length and returns the indices where the input vectors are different.    For a vector $\vc$ of length, say $p$, and an index set $\I \subset [p]$,  $\puncture_{\I}(\vc)$  denotes the pruned vector  that remains after removing the bits of $\vc$ corresponding to the indices listed in $\I$.  For a set $\C$ of vectors $\{ \vc_1,\ldots, \vc_m\}$,  $\puncture_{\I}(\C)$ denotes the set of pruned vectors $\{ \puncture_{\I}(\vc_1),\ldots, \puncture_{\I}(\vc_m)\}$.

\noindent{\em Security Parameters.} We denote the statistical security parameter by $\ssec$ and the cryptographic security parameter by $\csec$.  A negligible function in $\csec$ ($\ssec$) is denoted by $\negl(\csec)$ ($\negl(\ssec)$), while $\negl(\csec, \ssec)$ denotes a function which is negligible in both $\kappa$ and $\ssec$. A function $\negl(\cdot)$ is {\em negligible} if for every polynomial $p(\cdot)$ there exists a value $N$ such that for all $n>N$ it holds that $\negl(n)<\frac{1}{p(n)}$. We  write \ppt\ for probabilistic polynomial-time.

\noindent{\em Oblivious Transfers.}  For oblivious transfers, we denote the sender by $\sen$ and the receiver by $\rec$. In a $1$-out-of-$2$ OT on $\ell$ bit strings, $\sen$ holds two inputs $x_0,x_1$, each from $\bitset^\ell$ and $\rec$ holds a {\em choice bit} $b$. The output to $\rec$ is $x_b$ and no other party learns anything. We denote a $1$-out-of-$2$ OT on $\ell$ bit strings as $\ot{\ell}{ }{2}$. We denote a $1$-out-of-$n$ OT on $\ell$ bit strings as $\ot{\ell}{ }{n}$.   Finally,  we write $\ot{\ell}{m}{n}$ to denote $m$ instances of $\ot{\ell}{ }{n}$.  Similarly,  $\ot{\ell}{m}{2}$ denotes $m$ instances of $\ot{\ell}{ }{2}$.

\subsection{Walsh-Hadamard (WH) Codes}
Walsh-Hadamard  (WH)  code is a linear code over a binary alphabet $\F_2$ that maps messages of length $p$ to codewords of length $2^p$. We use WH code that maps messages of length $\log{\csec}$ to codewords of length $\csec$.  For $\vx \in \bitset^{\log{\csec}}$, $\WH(\vx)$ denotes the WH encoding of $\vx$ defined as  $\WH(\vx) \define (\vx \otimes \va)_{\va \in \bitset^{\log{\csec}}}$.  It is the $\csec$-bit string consisting of inner products of each $\log{\csec}$-bit string $\va$ with $\vx$. For each $\csec$, the WH code, denoted by $\C_{\WH}^\csec$ is defined as the set $\left\{\WH(\vx)\right\}_{\vx \in \bitset^{\log{\csec}}}$. Note that  $\C_{\WH}^\csec$ contains $\csec$ codewords each of length $\csec$ bits. Our OT extension protocol  relies on the following well-known property of WH codes.

\begin{theorem}
	The  distance of $\C_{\WH}^\csec$ is $\frac{\csec}{2}$ when $\csec$ is a power of $2$.
\end{theorem}

\subsection{Hash Function and Random Oracle Model}
\label{sec:RO}
We use a hash function  $H : \bitset^* \rightarrow \bitset^{\poly(\csec)}$ which we  model as a random oracle. Namely, we prove the security of our protocol assuming that $H$ implements a functionality $\FRAND$ which  for different inputs $x$, returns uniform random output values from the range of $H(x)$. 

\section{An Attack on \cite{KolesnikovK13} Protocol }
\label{sec:kk13}
The KK13 OT extension protocol  is known to provide the best communication complexity among the existing constructions when   the input length of the sender is `short'.  The protocol is proven to be secure against a semi-honest receiver and  a malicious sender. It was {\em not} known if the protocol is secure against malicious receiver. We show that the protocol is {\em insecure} against a malicious receiver. We prove this by giving an attack that can be mounted by a maliciously corrupt receiver to break the security of the sender. Our finding sets the stage for a maliciously secure OT extension in KK13 style which is the concern of this paper. Below we recall the KK13 OT extension protocol prior to presenting our attack. We also briefly recall its efficiency analysis from  \cite{KolesnikovK13}.

\subsection{KK13 OT Extension Protocol}\label{subsec:kk13}
The OT extension protocol constructs a $\ot{\ell}{m}{n}$   relying on an instance of $\ot{\csec}{\csec}{2}$.  We recall the simpler version of the protocol that reduces   $\ot{\ell}{m}{n}$  to  $\ot{m}{\csec}{2}$. It is well-known that $\ot{m}{\csec}{2}$ can be constructed from $\ot{\csec}{\csec}{2}$ with some additional cost.

Following the footstep of~\cite{IshaiKNP03}, KK13 OT extension allows the receiver to send an $m\cross \csec$ matrix column-wise to the sender using an instance of $\ot{m}{\csec}{2}$ where the sender acts as the receiver and vice versa.  In~\cite{IshaiKNP03} OT extension, the $i$th row of the transferred matrix allows the sender to create two pads for the two messages in the $i$th extended OT.  One of  the two pads is a function of  the sender's input bit vector to  $\ot{m}{\csec}{2}$ and thus is unknown to the receiver.  The other pad is completely known to the receiver.   The pad known to the receiver is used as the mask for the intended message of the receiver.  The above allows the receiver to unmask and learn its intended message for each extended OT but nothing more.  Going along the same line,  KK13 OT extension allows the sender to create  $n$ pads for the $n$ messages in the $i$th extended OT using the $i$th row of the transferred matrix. Much like IKNP, the receiver knows exactly one pad out of the $n$ pads and the pad it knows is in fact the mask for its intended message. All the remaining $n-1$ pads are function of the sender's input bit vector to  $\ot{m}{\csec}{2}$ and thus are unknown to the receiver.   The ability to generate $n$ masks instead of just $2$ from each of the rows of the transferred matrix  is achieved by cleverly incorporating WH codewords from $\C_\WH^\csec$ in each of the rows of the transferred $m \cross \csec$ matrix. The use of  $\C_\WH^\csec$ restricts the value of $n$ to be at most $\csec$.

The protocol uses WH code $\C_{\WH}^{\csec}$ that consists of $\csec$ codewords each of length $\csec$ denoted as $\vset{\vc}{1}{\csec}$. The receiver $\rec$ chooses two random $m\cross \csec$ matrices $\mB$ and $\mD$ such that $i$th row of matrix $\mE = \mB \xor \mD$ is $\vc_\subscript{r_i}$  where $r_i$ is the input of the receiver for the $i$th extended OT. On the other hand, the sender $\sen$ picks a $\csec$ bit length vector $\vs$ uniformly at random. The parties then interact via $\ot{m}{\csec}{2}$ reversing their roles. Namely, the sender $\sen$ acts as the receiver with input $\vs$ and the receiver $\rec$ acts as a sender with inputs $\{\vb^j, \vd^j \}_{j \in [\csec]}$. After the execution of $\ot{m}{\csec}{2}$, the sender holds an $m\cross \csec$ matrix $\mA$ such that the $i$th row of $\mA$ is the $i$th row of $\mB$ xored with the bitwise AND of $\vs$ and $\vc_{r_i}$, i.e. $\va_i =  \left(\vb_i \xor (\vs \band \vc_\subscript{r_i}) \right)$. With the $i$th row $\va_i$ of the matrix $\mA$, the sender creates $n$ pads for the $n$ messages in the $i$th extended OT as follows: $\big\{H\big(i, \va_i \xor (\vs \band \vc_j )\big)\big\}_{j \in [n]}$ where $H$ is a random oracle.  The $j$th pad will be used to blind the $j$th message of the sender in the $i$th extended OT. It is easy to note that the pad for the $r_i$th message is $H(i, \vb_i )$ (since $\va_i =  \left(\vb_i \xor (\vs \band \vc_{r_i}) \right)$) which the receiver can compute with the knowledge of $\mB$ matrix. For the $j$th message where $j$ is different from $r_i$,  the pad turns out to be $H\big(i, \vb_i \xor (\vs \band( \vc_\subscript{r_i} \xor\vc_j ))\big)$ where $\vc_\subscript{r_i}$ and $\vc_j$ are distinct codewords. Since the distance of WH code $\C_\WH^\csec$ is $\csec/2$, $\vc_{r_i}$ and $\vc_j$ are different at $\csec/2$ positions implying that $\csec/2$ bits of $\vs$ contribute to the input of the random oracle $H$. Since the vector $\vs$ is unknown to the receiver (recall that the sender picks $\vs$),   it is hard for an \ppt~receiver to retrieve the other pads making the protocol secure for a sender. The protocol of KK13 that realizes $\ot{\ell}{m}{n}$ given ideal access to $\ot{m}{\csec}{2}$ appears in Fig. \ref{fig:KK13OTEXT}.

\begin{boxfig*}{The KK13 OT Extension Protocol}{KK13OTEXT}
	\begin{center}
		\textbf{Protocol for $\ot{\ell}{m}{n}$  from  $\ot{m}{\kappa}{2}$}
	\end{center}
	\begin{mydescription}
		\item[--] \textbf{Input of $\sen$:} $m$ tuples $\big\{\vset{\vx}{i,1}{i,n} \big\}_{i \in [m]}$ of $\ell$ bit strings.
		
		\item[--] \textbf{Input of $\rec$:} $m$ selection integers $\vset{r}{1}{m}$ such that each $r_i \in [n]$.
		
		\item[--] \textbf{Common Inputs:} A security parameter $\csec$ such that $\csec \geq n$, and Walsh-Hadamard code ${\C}_{\WH}^\csec = \vset{\vc}{1}{\csec}$.
		
		\item[--] \textbf{Oracles and Cryptographic Primitives:} A random oracle $H : [m] \cross \bool^\csec \to \bool^\ell $. An ideal $\ot{m}{\csec}{2}$ primitive.
	\end{mydescription}
	\begin{enumerate}
		
		\item{\bf Seed OT Phase:} 
		
		\begin{enumerate}
			\item $\sen$ chooses $\vs \from \bool^\csec$ at random. 
			\item $\rec$ forms two $m \cross \csec$ matrices $\mB$ and $\mD$ in the following way:
			\begin{mydescription}
				\item[--] Choose $\vb_i, \vd_i \from \bool^\csec$ at random such that $\vb_i \xor \vd_i = \vc_\subscript{r_i}$.
				Let $\mE \define \mB \xor \mD$. Clearly $\ve_i = \vc_\subscript{r_i}$.
			\end{mydescription}
			
			\item $\sen$ and $\rec$ interact with $\ot{m}{\csec}{2}$ in the following way.  
			\begin{mydescription}
				\item[--] $\sen$ acts as {\em receiver} with input $\vs$.
				\item[--] $\rec$ acts as {\em sender} with input $\big\{(\vb^j, \vd^j) \big\}_{j \in [\csec]}$.
				\item[--] $\sen$ receives output $\{\va^j\}_{j \in [\csec]}$ and forms $m \cross \csec$ matrix $\mA$ with the $j$th column of $\mA$ as $\va^j$. Clearly  
				\begin{mydescription}
					\item[i.] $\va^j = \big(\vb^j \xor (s_j \band \ve^i)\big)$ and
					\item[ii.]  $\va_i = \big(\vb_i \xor (\vs \band \ve_i) \big) = \big(\vb_i \xor (\vs \band \vc_\subscript{r_i}) \big)$. 
				\end{mydescription}
			\end{mydescription}
		\end{enumerate}       
		\item{\bf OT Extension Phase:} 
		\begin{enumerate}
			\item For every $i \in [m]$, $\sen$ computes $\vy_{i,j} = \vx_{i,j} \xor H\big(i, \va_i \xor (\vs \band \vc_j )\big)$ and sends $\{\vy_{i,j}\}_{j \in [n]}$.
			
			\item For every $i \in [m]$, $\rec$ recovers $\vz_i = \vy_\subscript{i,r_i} \xor H(i,\vb_i)$.
		\end{enumerate}   
	\end{enumerate} 
\end{boxfig*}

It is easy to verify that the protocol is correct (i.e., $\vz_i = \vx_{i,r_i}$) when both parties follow the protocol. 

\subsection{An Attack}  \label{sec:kk13attack}
At the heart of the attack lies a clever way of manipulating the $\mE$ matrix (cf. Section~\ref{subsec:kk13}) which  should  contain  WH codewords in its rows in an honest execution.  Recall that the security of the sender lies in the fact that the distance of WH code $\C_\WH^\csec$ is $\csec/2$.  The pads for the messages that are not chosen as the output by the receiver, are the random oracle outputs of an input consisting of $\csec/2$ bits of $\vs$. Since the receiver $\rec$ does not know $\vs$, it cannot guess the pads too in polynomial time. So one way of breaking the privacy of the other inputs of the sender is to find out the bits of the vector  $\vs$. Our strategy allows the receiver to recover the $i$th bit of $\vs$ at the cost of two calls to the random oracle under the assumption that $\rec$ has apriori knowledge of its chosen input $\vx_\subscript{i,r_i}$ for the $i$th extended OT.  This is achieved by tweaking the rows of $\mE$ matrix which are codewords from $\C_\WH^\csec$ in an honest execution. Specifically, the $i$th row of $\mE$, $\ve_i$ is $\vc_\subscript{r_i}$ in an honest execution. It is now tweaked to a $\csec$-bit string that is same as $\vc_\subscript{r_i}$ in all the positions barring the $i$th position. Specifically, recall that a WH codeword $\vc_i$ from $\C_\WH^\csec$ is a $\csec$-length bit vector $(c_{i,1},\ldots, c_{i,\csec})$. We denote complement of a bit $b$ by $\ovA{b}$. Then the $i$th row $\ve_i$ of $\mE$ is set as $(c_{i,1},\ldots,\ovA{c_{i,i}},\ldots, c_{i,\csec})$. The matrix is tweaked as above for every $i$th row as long as $i \leq \csec$.  The rest of the rows in $\mE$ starting from $\csec$ to $m$ do not need to be tweaked. The matrix $\mE$ after tweaking is given below. We denote the tweaked matrix as $\ovA{\mE}$ and the tweaked rows as $\ovA{\vc_{r_i}}$ for $i\leq \csec$.

{
	\centering
	\fontsize{9}{9}
	\[
	\ovA{\mE}
	=
	\begin{bmatrix}
	\ovA{c_\subscript{r_1,1}}  & c_{\subscript{r_1,2}}               & \dots     & \dots          & \dots   & c_{\subscript{r_1,\csec}} \\
	c_{\subscript{r_2,1}}        & \ovA{c_{\subscript{r_2,2}}}          & \dots     & \dots          & \dots   & c_{\subscript{r_2,\csec}} \\
	\vdots          & \vdots                   & \ddots    & \vdots         & \ddots  & \vdots \\
	c_{\subscript{r_i,1}}        & c_{\subscript{r_i,2}}               & \dots     & \ovA{c_{\subscript{r_i,i}}} & \dots   & c_{\subscript{r_i,\csec}} \\
	\vdots          & \vdots                   & \ddots    & \vdots         & \ddots  & \vdots \\
	c_{\subscript{r_\csec,1}}        & c_{\subscript{r_\csec, 2}}               & \dots     & c_{\subscript{r_\csec, j}} & \dots   & \ovA{c_{\subscript{r_\csec,\csec}}} \\
	c_{\subscript{r_\csec,1}}        & c_{\subscript{r_\csec, 2}}               & \dots     & c_{\subscript{r_{(\csec+1)}, j}} & \dots   & c_{\subscript{r_{(\csec+1)},\csec}} \\
	\vdots          & \vdots                    & \ddots    & \vdots         & \ddots  & \vdots \\
	c_{\subscript{m,1}}        & c_{\subscript{m,2}}               & \dots     & \dots          & \dots   & c_{\subscript{m,\csec}}
	\end{bmatrix}
	=
	\begin{bmatrix}
	\ovA{\vc_\subscript{r_1}}   \\
	\ovA{\vc_\subscript{r_2}}   \\
	\vdots     \\
	\ovA{\vc_\subscript{r_i}}   \\
	\vdots     \\
	\ovA{\vc_\subscript{r_\csec}}  \\
	\vc_{r_\subscript{\csec +1} }  \\
	\vdots     \\
	\vc_{r_m}
	\end{bmatrix}
	\]%
}

When $\rec$ uses $\ovA{\mE}$ instead of $\mE$, the $i$th row of $\mA$ for $i \leq \csec$ will be $\ovA{\va_i} = \left(\vb_i \xor (\vs \band \ovA{\vc_\subscript{r_i}}) \right)$. The pad used to mask the $r_i$th message $\vx_\subscript{i,r_i}$ in $i$th extended OT  is: 
\begin{align*}
H\big(i,\ovA{\va_i} \xor (\vs \band \vc_{r_i} )\big) &= H\Big(i, \vb_i  \xor \big(\vs \band (\ovA{\vc_{r_i}} \xor \vc_{r_i} ) \big)\Big) \nonumber \\
&= H\Big(i, \vb_i \xor  (\vs \band  0^{i-1}10^{\csec -i}\Big) \nonumber \\
&= H\left(i, \vb_i \xor   0^{i-1}s_i0^{\csec -i}\right)  \nonumber
\end{align*}
Now note that the malicious receiver has cleverly made the pad used for $\vx_\subscript{i,r_i}$ a function of sole unknown bit $s_i$. With the knowledge of  its chosen input $\vx_\subscript{i,r_i}$ and the padded message $\vy_\subscript{i,r_i}$ that the receiver receives in the OT extension protocol, the malicious receiver $\rec$ recovers the value of the pad  by finding $\vy_\subscript{i,r_i} \xor \vx_\subscript{i,r_i}$. It further knows that $\vy_\subscript{i,r_i} \xor \vx_\subscript{i,r_i}$ is  same as $H\left(i, \vb_i \xor   0^{i-1}s_i0^{\csec -i}\right)$. Now two calls to the random oracle $H$ with inputs $\big\{(i, \vb_i \xor   0^{i-1}s_i0^{\csec -i})\big\}_{s_i \in \bitset}$ is sufficient to find the value of $s_i$. In the similar way, it can find entire input vector of the sender, $\vs$ with $2\csec$  number (polynomial in $\csec$) of calls to the random oracle breaking the  privacy of the sender completely.  The attack works in the version of KK13 that reduces   $\ot{\ell}{m}{n}$   to $\ot{\csec}{\csec}{2}$ without any modification. 

\subsection{Efficiency of \cite{KolesnikovK13}}
\label{kk13eff}
Since efficiency is the prime focus of this paper and we build an OT extension protocol in KK13 style secure against malicious adversaries, we recall the communication complexity of KK13 from~\cite{KolesnikovK13}. For complexity analysis we consider the version of KK13 that reduces   $\ot{\ell}{m}{n}$   to $\ot{\csec}{\csec}{2}$ (presented in Appendix D of~\cite{KolesnikovK13}) and requires less communication than the one  that reduces   $\ot{\ell}{m}{n}$  to  $\ot{\csec}{m}{2}$.    The communication complexity of KK13 OT extension producing $\ot{\ell}{m}{n}$ is $\Order(m(\csec + n\ell))$ bits.  

The best known semi-honest OT extension protocol before KK13 is IKNP protocol \cite{IshaiKNP03} which has a communication complexity  of $\Order(m(\csec + \ell))$ bits for  producing   $\ot{\ell}{m}{2}$ from  $\ot{\csec}{\csec}{2}$. To get $\ot{\ell}{m}{n}$ as the output from IKNP protocol, the efficient  transformation of  \cite{NaorP05} is used. The transformation generates  $\ot{\ell}{1}{n}$ from $\ot{\csec}{\log{n}}{2}$ with an additional (outside the execution of $\ot{\csec}{\log{n}}{2}$) communication cost of $\Order(\ell n)$ bits. 
This transformation can be repeated $m$ times to reduce $\ot{\ell}{m}{n}$ to $\ot{\csec}{m\log{n}}{2}$ with an additional communication cost of $\Order(\ell m n)$ bits.  So to get $\ot{\ell}{m}{n}$ as the output from IKNP protocol,  first   $\ot{\csec}{m\log{n}}{2}$  is produced via ~\cite{IshaiKNP03}  and then the reduction from $\ot{\ell}{m}{n}$ to  $\ot{\csec}{m\log{n}}{2}$ is used that requires  an additional communication cost of $\Order(\ell m n)$ bits. So the total communication turns out to be $\Order(m\log{n}\cdot(\csec + \csec) + \ell m n) = \Order(m(\csec \log{n} +n\ell ))$ bits. Now recall that  $n\leq \csec$, a restriction that comes from the KK13 OT extension (due to the fact that ${\C}_{\WH}^\csec$ contains $\csec$ codewords). Given this bound, as long as $\ell = \Omega(\log{n})$, KK13 OT extension gives better communication complexity than IKNP protocol.


\section{Actively Secure OT Extension for Short Secrets}
\label{malkk13}
We make the KK13 OT extension protocol secure against a malicious receiver by adding a consistency check that relies on linearity of WH code and adds a communication of $\Order(\ssec \log{\csec})$ bits irrespective of the number of  extended OTs. We  first discuss the properties of WH code relevant to us  for the correctness of the consistency check. We then discuss the check   and  our actively secure protocol. As we will see  the check involves an additional trick apart from the linearity of WH codes to achieve the claimed communication complexity.   We also describe the required ideal functionalities. 

\subsection{Randomized Linearity Testing}\label{linearity_test}
We focus on WH code that maps messages of length $\log{\csec}$ to codewords of length $\csec$. A WH codeword for a $\log{\csec}$-bit input  $\vx$ can be viewed as a truth table of a linear function $\lf_\vx: \bitset^{\log{\csec}} \rightarrow \bitset$ parametrised with $\vx$ where  $\lf_\vx(\va) =  \vx \otimes \va$. The WH codeword for $\vx$ can be defined as $\WH(\vx) \define (\lf_\vx(\va))_{\va \in \bitset^{\log{\csec}}}$. It is easy to note that $\lf_\vx(\va) = \lf_\va(\vx)$ for any $\va \in \bitset^{\log{\csec}}$. So we can rewrite the WH codeword for $\vx$ as $\WH(\vx) \define (\lf_\va(\vx))_{\va \in \bitset^{\log{\csec}}}$. It is also easy to note that $\lf_\va()$ is a linear function since $\lf_\va(\vx \xor \vy) =\lf_\va(\vx) \xor \lf_\va(\vy)$ for any $\vx$ and $\vy$ in $\bitset^{\log{\csec}}$. This implies that given codewords, say $\vc_\vx$ and $\vc_\vy$ corresponding to $\vx$ and $\vy$ respectively, the codeword for $\vx \xor \vy$ can be obtained as $\vc_\vx \xor \vc_\vy$.  In general, any linear combination of a set of WH codewords will lead to a WH codeword. On the other hand XOR of a codeword and a non-codeword will be a non-codeword.  We note that the above statements are true for pruned code $\puncture_\I(\C_\WH^\csec)$ for any $\I$ of size less than $\csec/2$.  The distance of $\puncture_\I(\C_\WH^\csec)$ is $\csec/2 - |\I|$ which is at least $1$.

In our OT extension protocol, we need to verify whether a set strings are individually valid WH codewords or not. In particular the number of strings to be verified is proportional to the number of extended OTs output by the OT extension protocol. In practice, it will be in the order of millions. Individual string testing may inflate the computation and the communication cost many-fold. We take the following route to bypass the efficiency loss. Given $\nu$ strings for validity verification, we  compress them to one string via linear combination taken using a uniform random vector of length $\nu$ and then check the compressed string {\em only} for validity. We show that the compression process ensures that the output string will be a non-codeword with probability at least $\frac{1}{2}$ if  the input set contains some non-codeword(s). Below we present the randomized linearity test for $\nu$ strings in Fig~\ref{fig:RAND_LIN_TEST_many} and its probability analysis in Theorem~\ref{th:RanLinTest}. 

\begin{boxfig}{A Randomized Linearity Test for Many Strings}{RAND_LIN_TEST_many}
	\begin{center}
		\textbf{Randomized Linearity Test for $\nu$ Strings}
	\end{center}
	\begin{description}
		\item[--] \textbf{Input:} $\nu$ $\csec$-bit strings $\vy_1,\ldots,\vy_\nu$.
		
		\item[--] \textbf{Output:} $\accept$ or $\reject$ indicating whether the strings $\vy_{\subscript{1}},\ldots,\vy_{\subscript{\nu}}$ passes the test or not. 
	\end{description}
	\begin{enumerate}
		
		\item{\bf Selection of Random Combiners:} Choose $\nu$ bits  $b_1,\ldots,b_\nu$ uniformly at random.

		\item{\bf The test:} Compute $\vy = \Xor_{i=1}^{\nu} b_{\subscript{i}} \band \vy_{\subscript{i}}$.  Output $\accept$ if $\vy$ is a valid WH codeword, output $\reject$ otherwise.
		
	\end{enumerate} 
\end{boxfig}

\begin{theorem}
	\label{th:RanLinTest}
	Assume that some of the $\nu$  $\csec$-bit strings $\vy_1,\ldots,\vy_\nu$ are not WH codewords. The randomized linearity test presented in Fig~\ref{fig:RAND_LIN_TEST_many} outputs $\reject$ with probability at least $\frac{1}{2}$.
\end{theorem}
\begin{proof}Without loss of generality, let $i_1,\ldots,i_\eta$ denote the indices of the input strings that are non-codewords.  That is, $\{i_1,\ldots,i_\eta\} \subseteq \{1,\ldots, \nu\}$ and $\vy_{\subscript{i_1}},\ldots,\vy_{\subscript{i_\eta}}$ are exactly the non-codeword strings among the set of $\nu$ input strings. It is easy to verify that  any linear combination of the remaining strings that are codewords will result in a codeword. So we concentrate on the linear combination that can result from the non-codewords $\vy_{\subscript{i_1}},\ldots,\vy_{\subscript{i_\eta}}$.  Let the uniform random bits used to find the linear combination of the non-codewords be $b_{\subscript{i_1}},\ldots,b_{\subscript{i_\eta}}$. There are $2^{\eta}$ possibilities in total for these $\eta$ bits which can be interpreted as numbers in the set  $\{0,\ldots, 2^{\eta}-1\}$.  We divide these $2^{\eta}$ strings or numbers in two sets, say $\A$ and $\B$. $\A$ and $\B$ consist of all the strings that corresponds to even and odd numbers respectively from $\{0,\ldots,2^{\eta}-1\}$. Clearly $|\A| = |\B| = 2^{\eta-1}$. We now show that at least $2^{\eta-1}$ strings lead to a non-codeword when they are used as linear combiners for the set of non-codewords $\vy_{\subscript{i_1}},\ldots,\vy_{\subscript{i_\eta}}$.  We prove our claim by showing  that for every element in set $A$, there exists at least one unique string that when used for linear combination of the non-codewords will  lead to a non-codeword. Consider a string $\vw$ from set $A$. We have two cases to consider: \\~
	\noindent {\sf (i)} $\vw$ when used as the linear combiner for $\vy_{\subscript{i_1}},\ldots,\vy_{\subscript{i_\eta}}$ yields a non-codeword. In this case $\vw$ itself is the string and element in $A$ that when used as the linear combiner for the non-codewords will  lead to a non-codeword.\\~
	\noindent {\sf (ii)} $\vw$ when used as a linear combiner for $\vy_{\subscript{i_1}},\ldots,\vy_{\subscript{i_\eta}}$ yields a codeword. Note that $\vw$ is a string that denotes an even number, say $p$ in  $\{0,\ldots,2^{\eta}-1\}$. The least significant bit of $\vw$ is a zero.  The string corresponding to $p+1$ will belong to the set $\B$ and will have the same form as $\vw$ except that the least significant bit will be $1$.  The linear combination of  $\vy_{\subscript{i_1}},\ldots,\vy_{\subscript{i_{\eta-1}}}$ with respect to $\vw$ is a codeword.  We exclude  $\vy_{\subscript{i_\eta}}$ from the list since the least significant bit of $\vw$ is zero.  Whereas $\vy_{\subscript{i_\eta}}$ is a non-codeword and will be included in the linear combination with respect to the string corresponding to $p+1$.

	Clearly, the string corresponding to $p+1$ will lead to a non-codeword as the linear combination of a codeword and a non-codeword always gives a non-codeword. We have shown that for every $\vw$ that leads to a codeword, there is a unique string in $\B$ that leads to a non-codeword. The mapping is one-to-one.  
	
	We can now conclude that at least half the possibilities of  $b_{\subscript{i_1}},\ldots,b_{\subscript{i_\eta}}$ leads to a non-codeword when used as a linear combiner. Since the linear combiners are chosen uniformly at random, the probability that the linear combination that will result from the non-codewords $\vy_{\subscript{i_1}},\ldots,\vy_{\subscript{i_\eta}}$ is a non-codeword is at least $\frac{1}{2}$.  Recall that  any linear combination of the remaining strings that are codewords will result in a codeword.  So the compressed string $\vy$ resulted from the linear combination of all the $\nu$ strings will be a non-codeword with at least $\frac{1}{2}$ probability too. 
\end{proof}

It is easy to note that the above theorem holds true for $\puncture_\I(\C_\WH^\csec)$  for any $\I$ of size less than $\csec/2$. So we get the following corollary. 
\begin{corollary}
	\label{cor:RanLinTestpruned}
	Let $\I \subset [\csec]$ be a set of size less than $\csec/2$.  Assume that some of the $\nu$  $\csec - |\I|$-bit vectors $\vy_1,\ldots,\vy_\nu$ are not pruned WH codewords. Then $\vy \not \in \puncture_{\I}(\C_\WH^\csec)$ with probability at least $\frac{1}{2}$ where $\vy = \Xor_{i=1}^{\nu} b_{\subscript{i}} \band \vy_{\subscript{i}}$ and the bits $b_1,\ldots,b_\nu$ are uniform random. 
\end{corollary}

\subsection{Functionalities}
We describe the ideal functionalities that we need. Below we present  an  OT functionality parameterized using three parameters $\ell$ that denotes the string length of the sender's inputs, $n$ that refers to $1$-out-of-$n$ OTs and $m$ that denotes the number of instances of the OTs.

\begin{boxfig}{The Ideal Functionality for $\ot{\ell}{m}{n}$}{FOT}
	\begin{center}
		\textbf{Functionality $\FOT{\ell}{m}{n}$}
	\end{center}
	$\FOT{\ell}{m}{n}$ interacts with $\sen$, $\rec$ and the adversary $\Sim$ and is parameterized by three parameters $\ell$ that denotes the string length of the sender's inputs, $n$ that refers to $1$-out-of-$n$ OTs and $m$ that denotes the number of instances of the OTs. 
	\begin{itemize}
		\item Upon receiving $m$ tuples $\big\{\vset{\vx}{i,1}{i,n} \big\}_{i \in [m]}$ of $\ell$ bit strings from $\sen$ and $m$ selection integers $\vset{r}{1}{m}$ such that each $r_i \in [n]$ from $\rec$, the functionality  sends $\big\{\vx_{i,r_i} \big\}_{i \in [m]}$ to $\rec$. Otherwise it aborts.
	\end{itemize}
\end{boxfig}

Next we present a functionality to generate  uniformly random common coins. 
\begin{boxfig}{The Ideal Functionality for generating random common coins}{FCOIN}
	\begin{center}
		\textbf{Functionality $\FCOIN$}
	\end{center}
	$\FCOIN$ interacts with $\sen$, $\rec$ and the adversary $\Sim$. 
	\begin{itemize}
		\item Upon receiving $(\coin,\ell)$ from both $\sen$ and $\rec$,  the functionality generates $\ell$ random bits, say $\vw$ and sends $\vw$ to both $\sen$ and $\rec$. Otherwise it aborts.
	\end{itemize}
\end{boxfig}

\subsection{The Protocol}

We now describe the protocol that realizes $\ot{\ell}{m}{n}$ given ideal access to $\ot{\csec}{\csec}{2}$. The protocol is similar to the protocol of KK13 (cf. Fig. \ref{fig:KK13OTEXT}), except that our protocol includes a consistency check for preventing $\rec$ from behaving maliciously and using non-codewords in matrix $\mE$. The check makes use of the Randomized Linearity Testing described in Section \ref{linearity_test}. It is trivial to see that Randomized Linearity Test alone doesn't suffice, since a malicious $\rec$ can provide some vector for the check independent from what he had used in the seed OTs. Thus we need a check to ensure that the vector provided by $\rec$ for the check is consistent with the vectors used in the seed OTs. We make use of the fact that if both $\sen$ and $\rec$ are honest, then we have $\va_i = \vb_i \xor \left(\ve_i \band \vs \right)$. A closer analysis of this expression gives a simple verification mechanism for a corrupt $\rec$. Namely, $\rec$ sends to $\sen$ a random linear combination  of the rows of $\mB$ and $\mE$, say $\vb$ and $\ve$ respectively, for a commonly agreed random linear combiner generated using a coin tossing protocol. $\sen$ then applies the same random linear combiner on the rows of  $\mA$ to obtain $\va$ and checks if $\vb$ and $\ve$ are consistent with $\vs$ and $\va$. Namely, whether $\va = \vb \xor  ( \ve \band \vs)$ holds or not. While the above check is simple, it requires  communication of $\csec$-bit vectors, namely $\vb$ and $\ve$. The communication is brought down to $\Order(\log{\csec})$ using a couple of tricks. First, a second level of compression function is applied on $\va$, $\vb$ and $  \ve \band \vs$ via xor on the bits of the individual vectors. This results in three bits $a$, $b$ and $p$ respectively from $\va$, $\vb$ and $ \ve \band \vs$.  Then the check is simply to verify if $a = b \xor p$. Notice that $\ve \band \vs$  can be perceived as  the linear combination of $\ve$ for random combiner $\vs$.   Since $\vs$ is privy to $\sen$, $\rec$ cannot compute the linear combination of $ \ve \band \vs$, namely $p$.   So $\rec$ sends across the index of the codeword that matches with $\ve$ and on receiving it $\sen$ computes $p$ after computing  $\ve \band \vs$. The index requires just $\log{\csec}$ bits as ${\C}_{\WH}^\csec$ consists of $\csec$ codewords. Thus our final consistency check needs communication of $\Order(\log{\csec})$ bits and a sequence of cheap xor operations. Lastly, the above check is repeated $\ssec$ times, where $\ssec$ denotes the statistical security parameter. We show that either a corrupt $\rec$ tweaks few positions of the codewords allowing error-correction  or it is caught. Either event takes place  with overwhelming probability. Looking ahead to the proof,  the former event allows the simulator to extract the inputs of corrupted $\rec$ and thereby making the real and the ideal world indistinguishable with high probability. Whereas, the protocol is aborted in both the real and ideal worlds when the latter event happens.  Our construction appears in Fig.~\ref{fig:KK13OTEXTmal}.

\begin{boxfig*}{Actively Secure OT Extension Protocol}{KK13OTEXTmal}
	\begin{center}
		\textbf{Protocol for $\ot{\ell}{m}{n}$  from  $\ot{\csec}{\kappa}{2}$}
	\end{center}
	\begin{mydescription}
		\item[--] \textbf{Input of $\sen$:} $m$ tuples $\big\{\vset{\vx}{i,1}{i,n} \big\}_{i \in [m]}$ of $\ell$ bit strings.
		
		\item[--] \textbf{Input of $\rec$:} $m$ selection integers $\vset{r}{1}{m}$ such that each $r_i \in [n]$.
		
		\item[--] \textbf{Common Inputs:} A security parameter $\csec$ such that $\csec \geq n$, and Walsh-Hadamard code ${\C}_{\WH}^\csec = \vset{\vc}{1}{\csec}$.
		
		\item[--] \textbf{Oracles, Cryptographic Primitives and Functionalities:} A random oracle $H : [m] \cross \bool^\csec \to \bool^\ell $ and a pseudorandom generator $G : \bool^\csec \to \bool^{m+\ssec}$. An ideal OT functionality $\FOT{\csec}{\csec}{2}$ and an ideal coin tossing functionality $\FCOIN$. 
	\end{mydescription}
	\begin{enumerate}
		
		\item{\bf Seed OT Phase:} 
		\begin{enumerate}
			\item $\sen$ chooses $\vs \from \bool^\csec$ at random. 
			\item $\rec$ chooses $\csec$ pairs of seeds $(\vk^0_j,\vk^1_j)$ each of length $\csec$.
			\item $\sen$ and $\rec$ interact with $\FOT{\csec}{\csec}{2}$ in the following way.  
			\begin{mydescription}
				\item[--] $\sen$ acts as {\em receiver} with input $\vs$.
				\item[--] $\rec$ acts as {\em sender} with input $\big\{(\vk^0_i,\vk^1_i) \big\}_{i \in [\csec]}$.
				\item[--] $\sen$ receives output $\{\vk^{s_i}_i\}_{i \in [\csec]}$.
			\end{mydescription}
			

		\end{enumerate}  
		\item{\bf OT Extension Phase I:} 
		\begin{enumerate}
			\item $\rec$ forms three $(m +\ssec) \cross \csec$ matrices $\mB$, $\mE$ and $\mD$ in the following way and sends $\mD$ to $\sen$:
			\begin{mydescription}
				\item[--] Set $\vb^j = G(\vk^0_j)$. 
				\item[--] Set  $\ve_i = \vc_\subscript{r_i}$ for $i \in [m]$. For $i \in [m+1, m+\ssec]$, set   $\ve_i$ to a randomly picked codeword from  ${\C}_{\WH}^\csec$.
				\item[--] Set  $\vd^j = \vb^j \xor G(\vk^1_j) \xor \ve^j$. 
			\end{mydescription}
			\item On receiving $\mD$, $\sen$ forms $(m +\ssec) \cross \csec$ matrix $\mA$ with the $j$th column  of $\mA$ set as $ \va^j = \left( s_j \band \vd^j \right) \xor G(\vk^{s_j}_j)$. Clearly,  {\sf (i)} $\va^j = \big(\vb^j \xor (s_j \band \ve^j)\big)$ and {\sf (ii)} $\va_i = \big(\vb_i \xor (\vs \band \ve_i) \big) = \big(\vb_i \xor (\vs \band \vc_\subscript{r_i}) \big)$. 
		\end{enumerate}
		
		\item{\bf Checking Phase:} 
		\begin{enumerate}
			\item $\sen$ and $\rec$ invoke $\FCOIN$ with $(\coin,\ssec(m+\ssec))$ and receives an $\ssec$ $(m + \ssec)$-length random bit vectors say $\vw^{(1)},\ldots,\vw^{(\ssec)}$. On receiving the vectors, the parties do the following for $l \in [\ssec]$:
			\begin{mydescription}
				\item[--]  $\rec$ computes $\vb^{(l)} = \Xor\limits_{i=1}^{m+\ssec} w_{\subscript{i}}^{(l)} \band \vb_{\subscript{i}}$, $\ve^{(l)} = \Xor\limits_{i=1}^{m+\ssec} w_{\subscript{i}}^{(l)} \band \ve_{\subscript{i}}$ and $b^{(l)} = \Xor\limits_{i=1}^{\csec} b_{\subscript{i}}^{(l)}$. It sends $b^{(l)}$ and $\alpha^{(l)}$ where $\ve^{(l)} = \vc_{\alpha^{(l)}}$ to $\sen$.
				\item[--] $\sen$ computes $\va^{(l)} = \Xor\limits_{i=1}^{m+\ssec} w_{\subscript{i}}^{(l)} \band \va_{\subscript{i}}$, $a^{(l)} = \Xor\limits_{i=1}^{\csec} a_{\subscript{i}}^{(l)}$, $\vp^{(l)} =  \vs \band \vc_{\alpha^{(l)}}$ and  $p^{(l)} = \Xor\limits_{i=1}^{\csec} p^{(l)}_i$.  It  aborts the protocol  if  $a^{(l)} \not = b^{(l)} \xor p^{(l)}$.
			\end{mydescription}
			
		\end{enumerate}
		\item{\bf OT Extension Phase II:} 
		\begin{enumerate}
			\item For every $i \in [m]$, $\sen$ computes $\vy_{i,j} = \vx_{i,j} \xor H\big(i, \va_i \xor (\vs \band \vc_j )\big)$ and sends $\{\vy_{i,j}\}_{j \in [n]}$.
			
			\item For every $i \in [m]$, $\rec$ recovers $\vz_i = \vy_\subscript{i,r_i} \xor H(i,\vb_i)$.
		\end{enumerate}   
	\end{enumerate} 
\end{boxfig*}

\subsubsection{Security}
The correctness of our protocol follows from the correctness of the KK13 protocol and the correctness of the consistency check. While the former is explained in Section~\ref{sec:kk13}, the latter is explained below. The linearly combined vectors $\ve^{(l)}$ for $l \in [\ssec]$ will be  valid codewords follows directly from the  linearity of WH code as mentioned in Section \ref{linearity_test}. When  $\rec$ is honest we have $\va_i = \vb_i \xor \left(\ve_i \band \vs \right)$ and $\vc_{\alpha^{(l)}} = \ve^{(l)}$ for $l \in [\ssec]$. Thus, for every $l \in [\ssec]$, 
\begin{align*}	
\va^{(l)} &= \Xor\limits_{i=1}^{m+\ssec} w_{\subscript{i}}^{(l)} \band \va_{\subscript{i}} = \Xor\limits_{i=1}^{m+\ssec} w_{\subscript{i}}^{(l)} \band \left[ \vb_i \xor \left(\ve_i \band \vs \right) \right] \\
&= \left[ \Xor\limits_{i=1}^{m+\ssec} w_{\subscript{i}}^{(l)} \band \vb_i \right] \xor \left[ \left( \Xor\limits_{i=1}^{m+\ssec} w_{\subscript{i}}^{(l)} \band \ve_i \right) \band \vs \right] \nonumber \\
&= \vb^{(l)} \xor (\ve^{(l)} \band \vs) \nonumber \\
a^{(l)}  &= \Xor\limits_{i=1}^{\csec} a_{\subscript{i}}^{(l)} =  \Xor\limits_{i=1}^{\csec} (b_{\subscript{i}}^{(l)} \xor  (s_i \band e_i^{(l)})) = b^{(l)} \xor p^{(l)}  \nonumber
\end{align*}
Now it is easy to verify that the protocol is correct (i.e., $\vz_i = \vx_{i,r_i}$) when both the parties follow the protocol.

We now move on to the security argument for our protocol. The original OT extension of~\cite{KolesnikovK13} provides security against a malicious $\sen$. Since our consistency check involves message communication from $\rec $ to $\sen$, it does not offer any new scope for a malicious sender to cheat. However, the check may reveal some information about $\rec$'s input. Recall that $\rec$'s input is encoded in the rows of matrix $\mE$ and during the check, a random linear combination of the rows of $\mE$ (where the combiner is known to $\sen$) is presented to $\sen$ for verification. The check is repeated for $\ssec$ times. To prevent information leakage on $\rec$'s input,  $\mE$ is padded with $\ssec$ extra rows consisting of random codewords. This ensures that the linear combination presented in an instance of the check will look random and will bear no information about the $m$ rows of $\mE$ that encode $\rec$'s input, unless the  bits of the random combiner corresponding to the padded $\ssec$ rows are zero. However, the probability of that happening is only $\frac{1}{2^{\ssec}}$.

A corrupt $\rec$ can cheat by not picking the rows of $\mE$  as  codewords. Our consistency check ensures an overwhelming probability for catching such a misconduct when `large' number of positions in the codewords are tweaked. If few positions are tweaked, then we show that the tweaked codewords are error-correctable with high probability allowing the simulator in the proof to extract input of the corrupt $\rec$. We now prove security formally. 

\begin{theorem}\label{otext_proof}
	The protocol in Fig.~\ref{fig:KK13OTEXTmal} securely realizes $\FOT{\ell}{m}{n}$ in the ($\FOT{\csec}{\csec}{2},\FCOIN$)-hybrid model.
\end{theorem}
\begin{proof} Our proof is presented in Universal Composability (UC) framework recalled briefly in Appendix~\ref{sec:UC}. To prove the security of our protocol, we  describe two simulators. The simulator $\Sim_\sen$ simulates the view of a corrupt sender and appears in Fig.~\ref{fig:SIM_MAL_SND}. On the other hand, the simulator  $\Sim_\rec$ simulates the view of a corrupt receiver and is presented in Fig.~\ref{fig:SIM_MAL_REC}. 
	\begin{boxfig*}{Simulator  $\Sim_\sen$ for Malicious Sender}{SIM_MAL_SND}
		\begin{center}
			\textbf{Simulator  $\Sim_\sen$  for $\sen$}
		\end{center}
		The simulator  plays the role of the honest $\rec$ and  simulates  each step of the protocol $\ot{\ell}{m}{n}$  as follows.  The communication of the $\Env$ with the adversary  $\Adv$ who corrupts $\sen$ is handled as follows: Every input value received by the simulator from $\Env$ is written on $\Adv$'s input tape.  Likewise, every output value written by $\Adv$ on its output tape is copied to the simulator's  output tape (to be read by the environment $\Env$). 
		\begin{enumerate}
			\item{\bf Seed OT Phase:}  On behalf of $\FOT{\csec}{\csec}{2}$, $\Sim_\sen$ receives $\vs$,  the input of $\sen$ to the functionality $\FOT{\csec}{\csec}{2}$. Next it picks $\csec$  PRG seeds $\vk_i$ each of length $\csec$ and sends $\big\{\vk_i\big\}_{i \in [\csec]}$ to $\sen$ on behalf of $\FOT{\csec}{\csec}{2}$.

			\item{\bf OT Extension Phase I:} $\Sim_\sen$ picks a  $(m +\ssec) \cross \csec$ matrix $\mD$ uniformly at random and sends to $\sen$.  It then computes matrix $\mA$ using the PRG seeds sent to $\sen$, $\vs$ and $\mD$. Namely, it sets  $ \va^j = \left( s_j \band \vd^j \right) \xor G(\vk_j)$.

			\item{\bf Checking Phase:}  On receiving $(\coin,\ssec(m+\ssec))$ from $\sen$ on behalf of $\FCOIN$, $\Sim_\sen$ sends $\ssec$ $(m + \ssec)$-length random bit vectors say $\vw^{(1)},\ldots,\vw^{(\ssec)}$.  For $l \in [\ssec]$, it then computes $\va^{(l)}$ and $a^{(l)}$ using  $\vw^{(l)}$ and $\mA$ just as an honest $\sen$ does. It chooses a random WH codeword, say $\ve^{(l)}$,  sets $\vb^{(l)} = \va^{(l)} \xor (\vs \band \ve^{(l)})$ and computes $b^{(l)}$ using $\vb^{(l)}$. Finally, it sends $\alpha^{(l)}$, the index of $\ve^{(l)}$ and $b^{(l)}$ to  $\sen$.

			\item{\bf OT Extension Phase II:} On receiving $\big\{\vset{\vy}{i,1}{i,n} \big\}_{i \in [m]}$ from $\sen$, $\Sim_\sen$ computes $\vx_{i,j} = \vy_{i,j} \xor H\big(i, \va_i \xor (\vs \band \vc_j )\big)$ for $1 \leq i \leq m$ and sends $\big\{\vset{\vx}{i,1}{i,n} \big\}_{i \in [m]}$ to functionality $\FOT{\ell}{m}{n}$ on behalf of $\sen$.		
		\end{enumerate} 
	\end{boxfig*}
	
	We now prove that $\Ideal_{\FOT{\ell}{m}{n}, \Sim_\sen,\Env}  \stackrel{c}{\approx}  \Real_{\ot{\ell}{m}{n}, \Adv,\Env}$ when $\Adv$ corrupts $\sen$.  In ($\FOT{\csec}{\csec}{2},\FCOIN$)-hybrid model, we note that the difference between the simulated and the real view lies in  $\mD$ matrix. In the simulated world, the matrix $\mD$ is a random matrix, whereas in the real world it is a pseudo-random matrix. The indistinguishability can be proved via a reduction to PRG security.  
	
	\begin{boxfig*}{Simulator $\Sim_\rec$ for Malicious Receiver}{SIM_MAL_REC}
		\begin{center}
			\textbf{Simulator $\Sim_\rec$ for $\rec$.}
		\end{center}
		The simulator  plays the role of the honest $\sen$ and  simulates  each step of the protocol $\ot{\ell}{m}{n}$  as follows.  The communication of the $\Env$ with the adversary  $\Adv$ who corrupts $\rec$ is handled as follows: Every input value received by the simulator from $\Env$ is written on $\Adv$'s input tape.  Likewise, every output value written by $\Adv$ on its output tape is copied to the simulator's  output tape (to be read by the environment $\Env$). 
		\begin{enumerate}
			\item{\bf Seed OT Phase:} On behalf of $\FOT{\csec}{\csec}{2}$, $\Sim_\rec$ receives the input of $\rec$ to the functionality, namely $\big\{(\vk^0_i,\vk^1_i) \big\}_{i \in [\csec]}$. 
			
			\item{\bf OT Extension Phase I:}  On receiving matrix $\mD$ from $\rec$, $\Sim_\rec$ computes $\mE$ using the knowledge of  $\big\{(\vk^0_i,\vk^1_i) \big\}_{i \in [\csec]}$.  That is, it computes $\mE$ as $\ve^i = G(\vk^0_i) \xor G(\vk^1_i) \xor \vd^i$, where $i \in [\csec]$.	 
			
			\item{\bf Checking Phase:}  On receiving $(\coin,\ssec(m+\ssec))$ from $\rec$ on behalf of $\FCOIN$, $\Sim_\rec$ sends $\ssec$ $(m + \ssec)$-length random bit vectors say $\vw^{(1)},\ldots,\vw^{(\ssec)}$ to $\rec$. Then $l \in [\ssec]$, it receives $b^{(l)}$ and $\alpha^{(l)}$ from $\rec$ and performs the consistency check honestly like an honest $\sen$. If the check fails, then it sends $\abort$ to $\FOT{\ell}{m}{n}$ on behalf of $\rec$  and halts.  If none of the check fails,  $\Sim_\rec$  computes $\ve^{(l)} = \Xor\limits_{i=1}^{m+\ssec} w_{\subscript{i}}^{(l)} \band \ve_{\subscript{i}}$ using the rows of $\mE$ and finds  $\Fault_i = \hdi(\ve^{(l)}, \vc_{\alpha^{(l)}})$ for $l \in [\ssec]$. It then computes $\Fault = \bigcup \limits_{l=1}^{\ssec}  \Fault_l$.  If $|\Fault| \geq \csec/2$ or there exists an index $i$ such that $\puncture_\Fault(\ve_i) \not \in \puncture_\Fault({\C}_{\WH}^\csec)$\footnote{Note that $\puncture_\Fault({\C}_{\WH}^\csec)$ consists of $\csec$ vectors with distance $\csec/2 - |\Fault|$ which is at least one when $|\Fault| < \csec /2$. This follows from the fact that  the distance of $\C_{\WH}^\csec$ is $\csec/2$.}, then  it sends $\abort$ to $\FOT{\ell}{m}{n}$. Otherwise,  $\Sim_\rec$ extracts the $i$th input of $\rec$ as $r_i$ where  $\puncture_\Fault(\ve_i) =  \puncture_\Fault(\vc_{r_i})$ for $i \in [m]$.

			\item{\bf OT Extension Phase II:}  $\Sim_\rec$  sends  the input of $\rec$, namely $\vset{r}{1}{m}$ (such that each $r_i \in [n]$) to  functionality $\FOT{\ell}{m}{n}$ on behalf of $\rec$. From $\FOT{\ell}{m}{n}$, it receives $\big\{\vx_{i,r_i} \big\}_{i \in [m]}$. It then  runs the protocol with $\rec$ using $\big\{\vx_{i,r_i} \big\}_{i \in [m]}$ and  $0^\ell$ for the unknown inputs $\big\{ \vx_{i,j} \big\}_{i \in [m] \wedge j \neq r_i}$. 
		\end{enumerate}
	\end{boxfig*}
	Next, we prove that $\Ideal_{\FOT{\ell}{m}{n}, \Sim_\rec,\Env}  \stackrel{c}{\approx}   \Real_{\ot{\ell}{m}{n}, \Adv,\Env}$ when $\Adv$ corrupts $\rec$ via a series  of hybrids. The output of each hybrid is always just the output of the environment $\Env$.   Starting with $\Hyb_0=   \Real_{\ot{\ell}{m}{n}, \Adv,\Env}$, we gradually make changes to define $\Hyb_1$ and $\Hyb_2$ as follows:
	
	\begin{mydescnoindent}
		\item[$\Hyb_1$:] Same as $\Hyb_0$, except that in the {\bf Checking Phase}, the protocol is aborted when the simulator  $\Sim_\rec$ fails to extract the input of $\rec$.
		\item[$\Hyb_2$:] Same as $\Hyb_1$, except that  the default value $0^\ell$ is substituted for the inputs $\big\{ \vx_{i,j} \big\}_{i \in [m] \wedge j \neq r_i}$.
	\end{mydescnoindent}
	
	Clearly,  $\Hyb_2 = \Ideal_{\FOT{\ell}{m}{n}, \Sim_\rec,\Env}$.	 Our proof will conclude, as we show that every two consecutive hybrids  are computationally indistinguishable. \\\\~
	\noindent $\Hyb_0  \stackrel{c}{\approx}  \Hyb_1$:   The difference between $\Hyb_0$ and $\Hyb_1$ lies in the condition on aborting the protocol. 
	In $\Hyb_0$ the protocol is aborted when   $a^{(l)} \not = b^{(l)} \xor p^{(l)}$ for some $l \in [\ssec]$  (cf. Fig.~\ref{fig:KK13OTEXTmal}). Whereas, in $\Hyb_1$ the protocol is aborted when either  the condition for abortion in $\Hyb_0$ is true or  the extraction fails. The latter implies that either $|\Fault| \geq \csec/2$ or there exist an index $i$ such that $\puncture_\Fault(\ve_i) \not \in \puncture_\Fault({\C}_{\WH}^\csec)$.  Let $\PassCheck$ denote the event of passing the consistency check for a corrupt $\rec$ who commits a non-codeword matrix $\mE$ in the seed OT phase.  Let $\DecodeFail$ denote the event of failed extraction of $\rec$'s input. Lastly, let $\Dist$ denote the event that $\Env$ distinguishes between $\Hyb_0  $ and  $ \Hyb_1$. Then, we have $\prob[\Dist~|~ \neg \PassCheck] = 0$ (since the execution aborts in both hybrids) and $\prob[\Dist~| ~\PassCheck] = \prob[\DecodeFail~| ~\PassCheck]$. So we have,
	\begin{align}\label{eq1}
	\prob[\Dist] &= \prob[\Dist~| ~\PassCheck] \cdot \prob[\PassCheck] + \prob[\Dist~|~ \neg \PassCheck] \cdot \prob[\neg \PassCheck] \\ \nonumber
	&= \prob[\DecodeFail~ |~ \PassCheck] \cdot \prob[\PassCheck] 
	\end{align}
	We now show that $\prob[\Dist] $ is negligible in $\csec$ and $\ssec$  because either the probability of passing the check is negligible or the probability of failure in extraction when check has passed is negligible. In other words, we show that $\prob[\PassCheck] \leq \negl(\csec,\ssec)$ when $|\Fault| \geq \csec /2$ and $\prob[\DecodeFail~ |~ \PassCheck]  \leq \negl(\csec,\ssec)$ otherwise. We capture the above in the following two lemmas. 
	\begin{lemma}\label{BigFault}
		$\prob[\Dist] \leq \max(\frac{1}{2^{|\Fault|}},  \frac{1}{2^\ssec})$, when   $|\Fault| \geq \csec /2$. 
	\end{lemma}
	\begin{proof}
		When $|\Fault| \geq \csec /2$, we note that  $\prob[\DecodeFail~ |~ \PassCheck] = 1$ as the extraction always fails.  Plugging the equality in Equation~\ref{eq1}, we get  $\prob[\Dist]  = \prob[\PassCheck]$. Next we conclude the proof by showing that $ \prob[\PassCheck] = \max( \frac{1}{2^{|\Fault|}}, \frac{1}{2^\ssec})$ which is negligible in $\csec$ and $\ssec$.

		Consider $l$th iteration of the check in {\bf Checking Phase}. Recall that $a^{(l)}$ at $\sen$'s end is computed as follows. First $\va^{(l)}$ is calculated as  $\va^{(l)} = \Xor\limits_{i=1}^{m+\ssec} w_{\subscript{i}}^{(l)} \band \va_{\subscript{i}}$ where $\va_i = \big(\vb_i \xor (\vs \band \ve_i) \big)$. Denoting $\vb^{(l)} = \Xor\limits_{i=1}^{m+\ssec} w_{\subscript{i}}^{(l)} \band \vb_{\subscript{i}}$ and $\ve^{(l)} = \Xor\limits_{i=1}^{m+\ssec} w_{\subscript{i}}^{(l)} \band \ve_{\subscript{i}}$, we have $\va^{(l)} = \vb^{(l)} \xor (\vs \band \ve^{(l)})$. Lastly, denoting $b^{(l)} = \Xor\limits_{i=1}^{\csec} b_{\subscript{i}}^{(l)}$ and $p =  \Xor\limits_{i=1}^{\csec} s_{\subscript{i}} \band e_i^{(l)}$,  we have $a^{(l)} = \Xor\limits_{i=1}^{\csec} a_{\subscript{i}}^{(l)} =  b^{(l)} \xor p^{(l)}$. Let  a corrupt $\rec$ sends the index $\alpha^{(l)}$. Let $\bar b^{(l)}$ denote the bit sent along with $\alpha^{(l)}$  and let  $\bar p^{(l)} =  \Xor\limits_{i=1}^{\csec} s_{\subscript{i}} \band c_i^{(l)}$ where $\vc_{\alpha^{(l)}} = [c_1^{(l)},\ldots,c_\csec^{(l)}]$.   Now the check passes when  $b^{(l)} \xor p^{(l)} =  \bar b^{(l)} \xor \bar p^{(l)}$. The equation implies  that
		\begin{eqnarray} \nonumber
		b^{(l)} \xor \bar b^{(l)} &=&  p^{(l)} \xor \bar p^{(l)}   = \left( \Xor\limits_{i=1}^{\csec} s_{\subscript{i}} \band e_i^{(l)} \right) \Xor \left( \Xor\limits_{i=1}^{\csec} s_{\subscript{i}} \band c_i^{(l)} \right) \\ \nonumber
		& = &  \Xor\limits_{i=1}^{\csec} s_{\subscript{i}} \band (e_i^{(l)} \xor c_i^{(l)})
		\end{eqnarray}
		Now note that the bits of $\vs$ corresponding to the indices {\em not} in $\Fault$ do not have any impact on the value of $\Xor\limits_{i=1}^{\csec} s_{\subscript{i}} \band (e_i^{(l)} \xor c_i^{(l)})$. So $2^{\csec -|\Fault|}$ possibilities of the vector $\vs$ will lead to passing the check. Since $\vs$ is  chosen uniformly at random and is a $\csec$-length bit vector, the probability that the chosen vector will hit one of the $2^{\csec -|\Fault|}$ possibilities is $\frac{2^{\csec -|\Fault|}}{2^\csec}$.   The probability of passing the check is thus $\frac{1}{2^{|\Fault|}}$. Another way of passing the check is to hit the value of $\bar b^{(l)}$ in all the $\ssec$ instances of the check so that the equalities $b^{(l)} \xor p^{(l)} =  \bar b^{(l)} \xor \bar p^{(l)} $ for $l \in [\ssec]$ hold good. The probability of passing the check in this way thus turns out to be  $\frac{1}{2^{\ssec}}$. This concludes the proof.
		
	\end{proof}
	
	\begin{lemma}\label{SmallFault}
		$\prob[\Dist] \leq   \frac{1}{2^\ssec}$, when   $|\Fault| < \csec /2$. 
	\end{lemma}
	\begin{proof}
		From Equation~\ref{eq1}, we get  the inequality  $\prob[\Dist]  \leq \prob[\DecodeFail~ |~ \PassCheck] $. We now show that $\prob[\DecodeFail~ |~ \PassCheck] \leq  \frac{1}{2^\ssec}$.  We note that when  $|\Fault| < \csec /2$, the reason for failure in extraction happens because some of the pruned rows of $\mE$ do not belong to the the pruned code  $\puncture_\Fault({\C}_{\WH}^\csec)$.  That is, there exists an index $i$ such that $\puncture_\Fault(\ve_i) \not \in \puncture_\Fault({\C}_{\WH}^\csec)$. Now the fact that the distance of $\C_{\WH}^\csec$ is $\csec/2$ and the number of indices that are pruned are strictly less that $\csec/2$ implies that $\puncture_\Fault({\C}_{\WH}^\csec)$ consists of $\csec$ vectors with distance $\csec/2 - |\Fault|$ which is at least one. Now Corollary~\ref{cor:RanLinTestpruned} implies that  if some of the pruned rows of $\mE$ do not belong to $ \puncture_\Fault({\C}_{\WH}^\csec)$, then $\puncture_\Fault(\ve^{(l)})$ belongs to $ \puncture_\Fault({\C}_{\WH}^\csec)$ with probability at most $1/2$.    Since $\ve^{(l)}$s are computed using independent and uniformly picked random linear combiners, at least one of $\puncture_\Fault(\ve^{(1)}),\ldots, \puncture_\Fault(\ve^{(\ssec)})$  do not belong to  $ \puncture_\Fault({\C}_{\WH}^\csec)$ with probability at least $1-\frac{1}{2^\ssec}$.  Recall that $\ve^{(l)}$ is computed using $\vw^{(l)}$. But since $\puncture_\Fault(\ve^{(l)}) = \puncture_\Fault(\vc_{\alpha^{(l)}})$ and $ \puncture_\Fault(\vc_{\alpha^{(l)}}) \in \puncture_\Fault({\C}_{\WH}^\csec)$ for all $l \in [\ssec]$, it implies that $\puncture_\Fault(\ve^i)$ for all $i \in [m+\ssec]$ belong to  $ \puncture_\Fault({\C}_{\WH}^\csec)$ with probability at least $1-\frac{1}{2^\ssec}$. So we have $\prob[\neg \DecodeFail~ |~ \PassCheck] \geq 1- \frac{1}{2^\ssec}$ which implies  $\prob[ \DecodeFail~ |~ \PassCheck] \leq  \frac{1}{2^\ssec}$
	\end{proof}\\~ 

	\noindent $\Hyb_1  \stackrel{c}{\approx}  \Hyb_2$: The difference between $\Hyb_1$ and $\Hyb_2$ lies in the values for the inputs $\big\{ \vx_{i,j} \big\}_{i \in [m] \wedge j \neq r_i}$. In $\Hyb_1$ these values are the real values of an honest $\sen$ whereas  in $\Hyb_2$ these are the default value $0^\ell$. The security in this case will follow from the random oracle assumption of $H$. We proceed in two steps. First, assume that the distinguisher of  $\Hyb_1$ and $\Hyb_2$ does not make any query to $H$. We show that the pads used to mask the unknown inputs of $\sen$ will be uniformly random and independent of each other due to random oracle assumption. Recall that the pads for masking $\big\{ \vx_{i,j} \big\}_{i \in [m] \wedge j \neq r_i}$ are $\big\{ H\big(i, \vb_i \xor (\vs \band ( \vc_\subscript{r_i} \xor \vc_j) )\big)\big\}_{i \in [m] \wedge j \neq r_i}$. Since ${\C}_{\WH}^\csec$ is a WH code, the Hamming weight of each vector in the set $\big\{( \vc_\subscript{r_i} \xor \vc_j)\big\}_{i \in [m] \wedge j \neq r_i}$ is at least $\csec/2$. Since $\vs$ is picked at random from $\bool^\csec$, each of the values in  $\big\{\vs \band ( \vc_\subscript{r_i} \xor \vc_j)\big\}_{i \in [m] \wedge j \neq r_i}$  is uniformly distributed over a domain of size at least $2^{\csec/2}$. Now random oracle assumption lets us conclude that the pads $\big\{ H\big(i, \vb_i \xor (\vs \band ( \vc_\subscript{r_i} \xor \vc_j) )\big)\big\}_{i \in [m] \wedge j \neq r_i}$ are random and independent of each other and thus provide information-theoretic blinding guarantee to the values $\big\{ \vx_{i,j} \big\}_{i \in [m] \wedge j \neq r_i}$.  
	
	Next,   following the standard of proofs in the random oracle model we allow the distinguisher to make polynomial (in $\csec$) number of adaptive queries to $H$. Clearly, if a distinguisher makes a query to $H$ on any of the values  $\big\{\big(i, \vb_i \xor (\vs \band ( \vc_\subscript{r_i} \xor \vc_j)\big)\big)\big\}_{i \in [m] \wedge j \neq r_i}$ that are used to mask the unknown inputs of $\sen$, then it can distinguish between the hybrids. Such queries are denoted as {\em offending queries}. As long as no offending query is made, each of these $m(n-1)$ offending queries is (individually) distributed uniformly at random over a domain of size (at least) $2^{\csec/2}$ and so the distinguisher's probability of hitting upon an offending query  remains the same as in the case he does not make any query at all to $H$.  So if the distinguisher makes $q$ queries, then it's probability of distinguishing only increases by a polynomial factor over $2^{-\csec/2}$. 
\end{proof}

Our security proof relies on random oracle assumption of $H$. However, as mentioned in \cite{KolesnikovK13}, the random oracle assumption can be replaced with a generalized notion of correlation-robustness \cite{IshaiKNP03} referred as $C$-correlation-robustness \cite{KolesnikovK13} in a straightforward way. 
For completeness, we recall the definition of $C$-correlation-robust hash functions below.
\begin{definition}[\cite{KolesnikovK13}] \label{def:CCRF}
	Let $C = \{\vc_1,\ldots,\vc_n\}$ be a set of $\csec$-bit strings such $n = \poly(\csec)$ and for any $j,k$ with $j \not = k$, the Hamming distance between $\vc_j$ and $\vc_k$ is $\Omega(\csec)$. Then a hash function $H: \bool^* \rightarrow \bool^{\ell(\csec)}$ is $C$-correlation-robust if for any polynomial $m(\csec)$ and any non-uniform \ppt\   distinguisher $\Adv$ provided with input $C$	has $\negl(\csec)$ probability of distinguishing  the following distributions: 
	\begin{mydescription}
		\item[--] $\Big\{ \big\{\big(j,k,H(i,\vb_i  \xor ( (\vc_j \xor \vc_k) \band \vs ))\big) \big\}_{i \in [m], j,k \in [n], j \not = k}\Big\}$ where each string in $\{\vb_i\}_{i \in [m]}$ is a $\csec$-bit  and $\vs$ is a $\csec$-bit string chosen uniformly at random and independent of   $\{\vb_i\}_{i \in [m]}$.
		\item[--] $U_{m(n-1)\ell}$; $U_{m (n-1)\ell}$ denotes uniform distribution over $\bool^{m(n-1)\ell}$.
	\end{mydescription}
\end{definition}	

\subsubsection{Efficiency}

The actively secure protocol incurs a communication of  $\Order(\csec^2)$ bits  in {\bf Seed OT Phase}. In {\bf OT Extension Phase I}, $\rec$  sends  $ \csec(m+\ssec)$ bits to $\sen$. In {\bf Checking Phase}, $\sen$ and $\rec$ invokes $\FCOIN$. We follow the implementation of~ \cite{KellerOS15} for $\FCOIN$ that generates $\ssec(m + \ssec)$ bits at one go and uses a pseudorandom function (PRF) and a PRG. Let $F_k: \bitset^{2\csec} \rightarrow \bitset^\csec$ be a keyed PRF with $k  \in \bitset^\csec$  be a uniform random string and $G: \bitset^\csec \rightarrow \bitset^{\ssec(m + \ssec)}$ be a PRG. Then $\FCOIN$ can be realized as follows: 

\begin{mydescription}
	\item $\rec$ generates and sends random $s_\rec \gets \bitset^\csec$  to $\sen$.
	\item $\sen$ generates and sends random $s_\sen \gets \bitset^\csec$ to $\rec$.
	\item Both parties compute $s = F_k(s_\sen,s_\rec)$ and output $\vw_1 || \vw_2\ldots || \vw_\ssec = G(s)$ where each $\vw_i \in \bitset^{(m+\ssec)}$. 
\end{mydescription} 

With the above implementation of $\FCOIN$, {\bf Checking Phase} incurs a communication of $\Order(\ssec\log{\csec})$. In {\bf OT Extension Phase II},  $\sen$ communicates $mn\ell $ bits to $\rec$.  So the total communication our protocol  is $\Order (\csec^2 + \csec (m+\ssec)  +  \ssec \log{\csec} + mn\ell)  = \Order(m(\csec + n \ell))$ (assuming $m$ is asymptotically bigger than $\csec$ and $\ssec$) bits which is same as that of KK13 OT extension. 

Computation-wise, {\bf Checking Phase} constitutes the additional work that our protocol does over KK13 protocol. The additional work involves   cheap xor and bit-wise multiplications.

\section{Empirical Results}
\label{sec:imple}
We compare our work with the existing protocols of  KK13 \cite{KolesnikovK13}, IKNP~\cite{IshaiKNP03}, ALSZ15~\cite{AsharovL0Z15} and NNOB~\cite{NielsenNOB12}  in terms of communication and runtime in LAN and WAN settings. The implementation of KOS~\cite{KellerOS15} is not available in the platform that we consider for benchmarking. As  per the claim made in KOS, the runtime of their OT extension bears an overhead of $5\%$ with respect to IKNP protocol both in LAN and WAN. The communication complexity of KOS is at least the complexity of IKNP. These results allow to get a clear idea on how KOS fares compared to our protocol. 

In any practical scenario, the computation is not the prime bottleneck, as computing power has improved a lot due to improvements in hardwares. The communication overhead is the main issue, and so most of the aforementioned protocols are aimed at improving the communication complexity. Our empirical results show that our proposed protocol performs way better than even the passively secure IKNP in terms of communication complexity when $\ot{}{}{n}$s with short input are expected outcomes. Since ALSZ15, NNOB and KOS are built upon IKNP, they lag behind our protocol in performance too. Though our prime focus is to improve the communication complexity, our protocol outperforms IKNP and the existing actively secure protocols in runtime both in LAN and WAN setting. We now detail the software, hardware and implementation specifications used in our empirical analysis before presenting our experimental findings. 

\noindent{\em Software Details.}
We build upon the OT extension code provided by the Encrypto group on github~\cite{OTCODE}. It contains the OT extension implementations of KK13, NNOB and ALSZ15 in C++, using the Miracl library for elliptic curve arithmetic. We build upon the KK13 code for our actively secure protocol. AES-128 has been used for the PRG instantiation and the random oracle has been implemented by the SHA-256 hash function.   

\noindent{\em Hardware Details.}
We have tested the code in a standard LAN network and a simulated WAN setting.  Our machine has 8 GB RAM and an Intel Core i5-4690 CPU with 3.5 GHz processor speed.  For WAN simulation, we used the tc tool of Linux, where we introduced a round trip delay of 100 milliseconds into the network, with a limited bandwidth of 20 Mbps.

\noindent{\em Implementation Details.}
\label{sec:impdet}
We discuss our choice of $m$, $n$ and $\ell$ denoting the number of extended output OTs, the number of inputs of $\sen$ in each extended OT and the bit length of $\sen$'s input respectively. In other words, we refer to the parameters $m,n$ and $\ell$ of  $\ot{\ell}{m}{n}$.  Recall that as long as the input length of $\sen$, namely $\ell$ satisfies the relation $\ell = \Omega(\log{n})$, theoretically KK13 OT extension (and our proposed OT extension) gives better communication complexity for producing  $\ot{\ell}{m}{n}$ than IKNP protocol and its variants (cf. Section~\ref{kk13eff}).  For benchmarking, we take two approaches.   First, we fix $n =16$ and $\ell = 4 (= \log{16})$ and experiment on the following values of $m$:   $1.25\times10^5$, $2.5\times10^5$, $5\times10^5$ and  $1.25\times10^6$. Next, we fix $m$ to a value and vary $n$ from $8$ to $256$ in the powers of $2$.  The value of $\ell$ for each choice of $n$ is set as $\log{n}$. Our protocol and  KK13 directly generate OTs of type $\ot{\ell}{}{n}$ whereas  IKNP, ALSZ15 and NNOB generate OTs of type $\ot{1}{}{2}$. To compare with IKNP, ALSZ15 and NNOB,  we convert the output OTs of these protocols, namely $\ot{1}{}{2}$ to   $\ot{\ell}{}{n}$ using the efficient  transformation of  \cite{NaorP05} (cf. Section~\ref{kk13eff}). 

To obtain a computational security guarantee of $2^{-128}$, while KK13 and our protocol need $256$ seed OTs, IKNP, NNOB and ALSZ15 need  $128$, $342$ and respectively $170$ seed OTs.  Among these, except IKNP and KK13, the rest are maliciously secure. To obtain a statistical security guarantee of $2^{-40}$ against a malicious receiver, ALSZ15 and NNOB need $380$ checks whereas we need $96$ checks.

We follow the approach of ALSZ15 implementation and perform the OT extension in batches of $2^{16}$ in sequential order. For each batch, the sender and the receiver perform the seed OTs, participate in a coin tossing protocol, perform the checks and finally obtain the output.   We use one thread in the sender as well as in the receiver side  in order to calculate the upper bound on the computation cost. However our code is compatible with multiple threads where each thread can carry out a batch of OTs. 
Lastly, our seed OT implementation relies on the protocol of~\cite{PeikertVW08}.

\subsection{Performance Comparison}
Since we build upon KK13 protocol, we first display  in Table~\ref{tab:overkk1}-\ref{tab:overkk2} the overhead (in \%)  of our protocol compared to KK13. Notably, the communication overhead lies in the range $\mathbf{0.011}\% \mbox{-} \mathbf{0.028}\%$. Table~\ref{tab:overkk1} shows that for large enough number of extended OTs, the runtime overhead of our protocol over KK13 ranges between $4\mbox{-}6\%$ for both LAN and WAN.  Table \ref{tab:overkk2} demonstrates that  the runtime overheads drop below $2\%$ when  in addition the number of inputs of the sender in the extended OTs is large enough.

Next, our empirical results are shown in Table \ref{tab:compareall}-\ref{tab:compareallvaryn}  and  Fig. \ref{fig:compareall}-\ref{fig:compareallvaryn}.  First, we discuss our results in Table \ref{tab:compareall} and Fig. \ref{fig:compareall} where we vary $m$. Next, we focus on the results displayed in Table \ref{tab:compareallvaryn} and Fig. \ref{fig:compareallvaryn} where we vary $n$. In both the case studies, our protocol turns out to be the best choice among the actively secure OT extensions and second best overall closely trailing KK13 which is the overall winner.  Communication complexity wise, our protocol is as good as KK13  and is way better than the rest. The empirical results are in concurrence with  the theoretical $\log{n}$ improvement of KK13 and our protocol over IKNP (and its variants).




\begin{table}[h!]
	\centering
	\begin{minipage}{0.48\textwidth}
		\caption{\footnotesize{Runtime and Communication Overhead (in \%) of Our protocol over KK13 for producing $\ot{4}{m}{16}$.}}
		\centering
		\resizebox{.6\textwidth}{!}{
			\begin{tabular}{ l |c| c| c } 
				\toprule
				\multirow{4}{*}{$m$}  
				& \multicolumn{2}{c|}{Runtime} & \multirow{2}{*}{Communication}\rule{0pt}{3ex}  \\ 		
				\hhline{~--~}
				& LAN  & WAN  & \rule{0pt}{3ex}  \\ 		 			
				\midrule
				$1.25\times10^5$ 	& 6.48	& 9.27	& 0.012\\ 
				$2.5\times10^5$ 	& 6.33	& 8.76	& 0.012\\
				$5\times10^5$ 		& 5.88	& 7.09	& 0.012\\
				$1.25\times10^6$ 	& 3.78	& 5.72	& 0.028\\ [1ex]
				\bottomrule
			\end{tabular}
		}	
		\label{tab:overkk1}
	\end{minipage}
	\hfill
	\begin{minipage}{0.48\textwidth}
		\caption{\footnotesize{Runtime and Communication Overhead (in \%) of Our protocol over KK13 for producing  $\ot{\log{n}}{10^6}{n}$.}}
		\centering
		\resizebox{.6\textwidth}{!}{
			\begin{tabular}{ l |c| c| c } 
				\toprule
				\multirow{4}{*}{$n$ } 
				& \multicolumn{2}{c|}{Runtime} & \multirow{2}{*}{Communication}\rule{0pt}{3ex}  \\ 		
				\hhline{~--~}
				& LAN  & WAN  & \rule{0pt}{3ex}  \\ 		 			
				\midrule
				$8$ 	& 6.16	& 11.77	& 0.011\\ 
				$16$ 	& 4.13	& 6.6	& 0.012\\
				$32$ 	& 4.5	& 2.29	& 0.013\\
				$64$ 	& 3.65	& 1.81	& 0.014\\
				$128$ 	& 1.24	& 1.18	& 0.015\\
				$256$ 	& 0.58  & 0.8	& 0.015\\ [1ex] 
				\bottomrule
			\end{tabular}
		}
		\label{tab:overkk2}
	\end{minipage}
\end{table}

\noindent{\em Performance Comparison for varied $m$ values.} The results in Table \ref{tab:compareall} reflects that  KK13 is the best performer  in terms of communication as well as runtime in LAN and WAN. Our actively secure protocol is the second best closely trailing KK13. Our protocol has communication overhead of  $0.012\mbox{-}0.028 \%$ over KK13, while IKNP, ALSZ15 and NNOB have overheads of  $79.7\mbox{-}84\%$, $249\%$ and $352\%$ respectively.  Noticeably,  we observe that the cost for generating $5\times10^5$ $\ot{}{}{16}$ using our protocol is less than the cost of generating $1.25\times10^5$ $\ot{}{}{16}$ using NNOB. Similarly  the cost of generating $1.25\times10^6$ $\ot{}{}{16}$ using our protocol is $71.6\%$ of the cost of generating $5\times10^5$ $\ot{}{}{16}$ using ALSZ15. 

In LAN setting, the overheads over KK13 vary in the range of  $3.78\mbox{-}6.48 \%$,  $11\mbox{-}17.60\%$, $14.4\mbox{-}22.7\%$ and $14.5\mbox{-}20.8\%$ respectively for our protocol, IKNP, ALSZ15 and NNOB.  The similar figures in WAN setting are $5.72\mbox{-}9.26 \%$, $16\mbox{-}22.6\%$,  $35.3\mbox{-}39\%$ and $24.1\mbox{-}29.1\%$ respectively for our protocol,  IKNP, ALSZ15 and NNOB. A pictorial representation  is shown in  Fig. \ref{fig:compareall}.

\begin{table*}[h!]	
	\caption{\footnotesize Performance Comparison of various OT extensions for producing $\ot{4}{m}{16}$.}
	\centering
	\fontsize{20}{20}\selectfont
	\resizebox{1\textwidth}{!}{
		\begin{tabular}{| c | c | c | c | c | c | c | c | c | c | c | c | c | c | c | c |} 
			\toprule
			\multirow{2}{*}{$m$} & \multicolumn{5}{c |}{Runtime in LAN (in sec)} & \multicolumn{5}{c}{Runtime in WAN (in sec)} & \multicolumn{5}{| c |}{Communication (in MB)} \\
			\hhline{~---------------}
			& {\bf KK13} & {\bf This paper} & {\bf IKNP} &{\bf ALSZ15} & {\bf NNOB} & {\bf KK13} & {\bf This paper} & {\bf IKNP} &{\bf ALSZ15} & {\bf NNOB} & {\bf KK13} & {\bf This paper} & {\bf IKNP} &{\bf ALSZ15} & {\bf NNOB} \rule{0pt}{3ex}  \\
			\midrule
			&&&&&&&&&&&&&&& \\
			$1.25\times10^5$	& 02.16	& 02.30 & 02.54 & 02.65 & 02.61 & 13.38	& 14.62 & 16.40 & 18.10 & 16.90	& 04.77	& 04.77 & 08.66 & 16.67 & 21.60 \\ 
			$2.5\times10^5$     & 04.23 & 04.50 & 04.88 & 05.26 & 05.05 & 24.32 & 26.45 & 29.26 & 33.80 & 31.40 & 09.54 & 09.54 & 17.15 & 33.33 & 43.21 \\
			$5\times10^5$ 	    & 08.50 & 09.00 & 09.78 & 10.04 & 10.10 & 47.39 & 50.75 & 56.9 & 65.00 & 60.40 & 19.08 & 19.08 & 34.79 & 66.62 & 86.39 \\
			$1.25\times10^6$ 	& 21.68 & 22.50 & 24.07 & 24.81 & 24.84 & 115.34 & 121.94 & 133.81 & 158.60 & 143.20 & 47.69 & 47.70 & 87.74 &	166.54 & 215.95 \\ 
			\bottomrule
		\end{tabular}
	}
	\label{tab:compareall}
\end{table*}

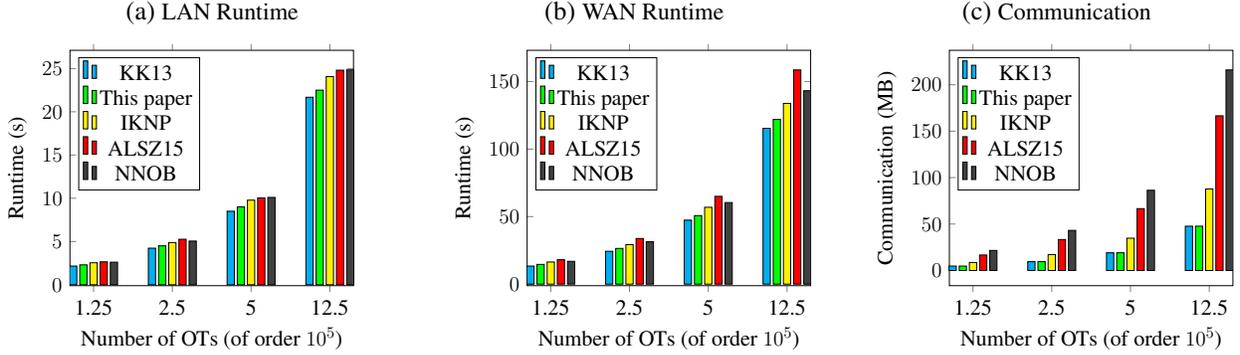
\begin{figure*}[htb!]
	\caption{\footnotesize Performance Comparison of various OT extensions for producing $\ot{4}{m}{16}$.}
	\centering
	\begin{subfigure}{.34\textwidth}
		\caption{\footnotesize LAN Runtime}
		\begin{tikzpicture} [scale = 0.55]
		\fontsize{15 pt}{\baselineskip}
		\begin{axis}[
		symbolic x coords={1.25, 2.5, 5, 12.5},
		xtick=data,
		ybar=2pt,
		samples=2,
		domain=1:2,
		ylabel=Runtime (s),
		xlabel=Number of OTs (of order $10^5$),
		ylabel near ticks,
		x label style={ below=3mm},
		xtick=data,
		bar width=5,
		legend pos=north west
		]

		\addplot[ybar,fill=cyan] coordinates {
			(1.25,  2.16)
			(2.5,  4.23)
			(5,  08.50)
			(12.5,  21.68)
		};
		
		\addplot[ybar,fill=green] coordinates {
			(1.25,  2.30)
			(2.5,  4.5)
			(5,  9)
			(12.5,  22.50)
		};

		\addplot[ybar,fill=yellow] coordinates {
			(1.25,  2.54)
			(2.5,  4.88)
			(5,  9.78)
			(12.5,  24.07)
		};
		
		\addplot[ybar,fill=red] coordinates {
			(1.25,  2.65)
			(2.5,  5.26)
			(5,  10.04)
			(12.5,  24.81)
		};
		
		\addplot[ybar,fill=darkgray] coordinates {
			(1.25,  2.61)
			(2.5,  5.05)
			(5,  10.10)
			(12.5,  24.84)
		};
		\legend{KK13, This paper, IKNP, ALSZ15, NNOB}
		\end{axis}
		\end{tikzpicture}
		\label{fig:runlanall}
	\end{subfigure}%
	\begin{subfigure}{.34\textwidth}
		\caption{\footnotesize WAN Runtime}
		\centering
		\begin{tikzpicture}[scale = 0.55]
		\fontsize{15 pt}{\baselineskip}
		\begin{axis}[
		symbolic x coords={1.25, 2.5, 5, 12.5},
		xtick=data,
		ybar=2pt,
		samples=2,
		domain=1:2,
		ylabel=Runtime (s),
		xlabel=Number of OTs (of order $10^5$),
		ylabel near ticks,
		x label style={ below=3mm},
		xtick=data,
		bar width=5,
		legend pos=north west
		]

		\addplot[ybar,fill=cyan] coordinates {
			(1.25,  13.38)
			(2.5,  24.32)
			(5,  47.39)
			(12.5,  115.34)
		};
		
		\addplot[ybar,fill=green] coordinates {
			(1.25,  14.62)
			(2.5,  26.45)
			(5,  50.75)
			(12.5,  121.94)
		};
		
		\addplot[ybar,fill=yellow] coordinates {
			(1.25,  16.40)
			(2.5,  29.26)
			(5,  56.9)
			(12.5,  133.81)
		};
		
		\addplot[ybar,fill=red] coordinates {
			(1.25,  18.10)
			(2.5,  33.80)
			(5,  65)
			(12.5,  158.6)
		};
		
		\addplot[ybar, fill=darkgray] coordinates {
			(1.25,  16.90)
			(2.5,  31.40)
			(5,  60.4)
			(12.5,  143.2)
		};
		\legend{KK13, This paper, IKNP, ALSZ15, NNOB}
		\end{axis}
		\end{tikzpicture}
		\label{fig:runwanall}
	\end{subfigure}%
	\begin{subfigure}{.34\textwidth}
		\caption{\footnotesize Communication}
		\centering
		\begin{tikzpicture}[scale = 0.55]
		\fontsize{15 pt}{\baselineskip}
		\begin{axis}[
		symbolic x coords={1.25, 2.5, 5, 12.5},
		xtick=data,
		ybar=2pt,
		samples=2,
		domain=1:2,
		ylabel=Communication (MB),
		xlabel=Number of OTs (of order $10^5$),
		ylabel near ticks,
		x label style={ below=3mm},
		xtick=data,
		bar width=5,
		legend pos=north west
		]

		\addplot[ybar,fill=cyan] coordinates {
			(1.25,  4.77)
			(2.5,  9.54)
			(5,  19.08)
			(12.5,  47.69)
		};
		
		\addplot[ybar,fill=green] coordinates {
			(1.25,  4.77)
			(2.5,  9.54)
			(5,  19.08)
			(12.5,  47.70)
		};
		
		\addplot[ybar,fill=yellow] coordinates {
			(1.25,  8.66)
			(2.5,  17.15)
			(5,  34.79)
			(12.5,  87.74)
		};
		
		\addplot[ybar,fill=red] coordinates {
			(1.25,  16.67)
			(2.5,  33.33)
			(5,  66.62)
			(12.5,  166.54)
		};

		\addplot[ybar,fill=darkgray] coordinates {
			(1.25,  21.60)
			(2.5,  43.21)
			(5,  86.39)
			(12.5,  215.95)
		};
		\legend{KK13, This paper, IKNP, ALSZ15, NNOB}[scale = 0.2]
		\end{axis}			
		\end{tikzpicture}			
		\label{fig:commall}
	\end{subfigure}
	\label{fig:compareall}
\end{figure*}

\noindent{\em Performance Comparison for varied $n$ values.} 
Here we set $m = 5\times 10^4$ and vary $n$ from $8$ to $256$ in the powers of $2$. Similar to the previous case study,    KK13 turns out the best performer  here as well (cf. Table \ref{tab:compareallvaryn}).  Our protocol is the second best closely trailing KK13. Our protocol has communication overhead of  $0.011\mbox{-}0.028 \%$ over KK13, while IKNP, ALSZ15 and NNOB have overheads of  $74.5\mbox{-}384.6\%$, $239.5\mbox{-}549.4\%$ and $341.8\mbox{-}652.9\%$ respectively. In LAN setting, the overheads over KK13 vary in the range of  $3.5\mbox{-}8.8 \%$,  $12.2\mbox{-}263.8\%$, $22.2\mbox{-}282.7\%$ and $20.5\mbox{-}267.5\%$ respectively for our protocol, IKNP, ALSZ15 and NNOB. The similar figures in WAN setting are $13.7\mbox{-}20.2 \%$, $27.1\mbox{-}77.4 \%$, $36.7\mbox{-}106.3\% $ and $30.9\mbox{-}86.6\%$ respectively for our protocol,  IKNP, ALSZ15 and NNOB. A pictorial representation  is shown in  Fig. \ref{fig:compareallvaryn}.




\begin{table*}[h!]	
	\caption{\footnotesize Performance Comparison of various OT extensions for producing  $\ot{\log{n}}{5 \cross 10^4}{n}$.}
	\centering
	\fontsize{20}{20}\selectfont
	\resizebox{1\textwidth}{!}{
		\begin{tabular}{| c | c | c | c | c | c | c | c | c | c | c | c | c | c | c | c |} 
			\toprule
			\multirow{2}{*}{$n$} & \multicolumn{5}{c |}{Runtime in LAN (in sec)} & \multicolumn{5}{c}{Runtime in WAN (in sec)} & \multicolumn{5}{| c |}{Communication (in MB)} \\
			\hhline{~---------------}
			& {\bf KK13} & {\bf This paper} & {\bf IKNP} &{\bf ALSZ15} & {\bf NNOB} & {\bf KK13} & {\bf This paper} & {\bf IKNP} &{\bf ALSZ15} & {\bf NNOB} & {\bf KK13} & {\bf This paper} & {\bf IKNP} &{\bf ALSZ15} & {\bf NNOB} \rule{0pt}{3ex}  \\
			\midrule
			&&&&&&&&&&&&&&& \\
			$8$	 & 0.70	 & 0.73	 & 0.79	 & 0.86	 & 0.85	 & 5.06	 & 6.08	 & 7.38	 & 7.67	 & 7.41	 & 1.43	 & 1.43	 & 2.5	 & 4.87	 & 6.34 \\ 
			$16$ & 0.96	 & 1.01	 & 1.15	 & 1.23	 & 1.17	 & 6.70	 & 7.72	 & 8.52	 & 9.15	 & 8.77	 & 1.91	 & 1.91	 & 3.53	 & 6.69	 & 8.65 \\
			$32$  & 1.22	 & 1.28	 & 1.48	 & 1.64	 & 1.58	 & 7.46  & 8.53  & 9.10	 & 10.2	 & 9.86	 & 2.39	 & 2.39	 & 4.89	 & 8.81	 & 11.29\\
			$64$ & 1.33	 & 1.45	 & 2.26	 & 2.37	 & 2.36	 & 8.88	 & 10.31	 & 11.29	 & 12.97	 & 12.64	 & 2.86	 & 2.86	 & 7.01	 & 11.7	 & 14.68 \\ 
			$128$ & 1.50	 & 1.63	 & 3.61	 & 3.94	 & 3.71	 & 9.56	 & 11.05	 & 14.84	 & 16.63	 & 15.16	 & 3.34	 & 3.34	 & 10.85	 & 16.36	 & 19.8 \\ 
			$256$ & 1.75	 & 1.89	 & 6.37	 & 6.70	 & 6.43	 & 10.96	 & 12.46	 & 19.43	 & 22.6	 & 20.45	 & 3.82	 & 3.82	 & 18.5	 & 24.8	 & 28.75 \\ [1ex] 
			\bottomrule
		\end{tabular}
	}
	\label{tab:compareallvaryn}
\end{table*}

\begin{figure*}[htb!]
	\caption{\footnotesize Performance Comparison of various OT extensions for producing  $\ot{\log{n}}{5 \cross 10^4}{n}$.}
	\centering
	\begin{subfigure}{.34\textwidth}
		\caption{\footnotesize LAN Runtime}
		\begin{tikzpicture} [scale = 0.55]
		\fontsize{12 pt}{\baselineskip}
		\begin{axis}[
		symbolic x coords={8, 16, 32, 64, 128, 256},
		xtick=data,
		ybar=.5pt,
		samples=2,
		domain=1:2,
		ylabel=Runtime (s),
		xlabel=Value of $n$,
		ylabel near ticks,
		x label style={ below=3mm},
		xtick=data,
		bar width=5,
		legend pos=north west
		]

		\addplot[ybar,fill=cyan] coordinates {
			(8, 0.704)
			(16, 0.956)
			(32, 1.218)
			(64, 1.332)
			(128, 1.500)
			(256, 1.750)
			
		};
		
		\addplot[ybar,fill=green] coordinates {
			(8, 0.729)
			(16, 1.008)
			(32, 1.285)
			(64, 1.450)
			(128, 1.630)
			(256, 1.891)
			
		};

		\addplot[ybar,fill=yellow] coordinates {
			(8, 0.790)
			(16, 1.149)
			(32, 1.478)
			(64, 2.263)
			(128, 3.608)
			(256, 6.368)
			
		};
		
		\addplot[ybar,fill=red] coordinates {
			(8, 0.860)
			(16, 1.226)
			(32, 1.640)
			(64, 2.369)
			(128, 3.938)
			(256, 6.698)
			
		};
		
		\addplot[ybar,fill=darkgray] coordinates {
			(8, 0.848)
			(16, 1.173)
			(32, 1.578)
			(64, 2.361)
			(128, 3.713)
			(256, 6.432)
			
		};
		\legend{KK13, This paper, IKNP, ALSZ15, NNOB}
		\end{axis}
		\end{tikzpicture}
		\label{fig:runlanallvaryn}
	\end{subfigure}%
	\begin{subfigure}{.34\textwidth}
		\caption{\footnotesize WAN Runtime}
		\centering
		\begin{tikzpicture}[scale = 0.55]
		\fontsize{12 pt}{\baselineskip}
		\begin{axis}[
		symbolic x coords={8, 16, 32, 64, 128, 256},
		xtick=data,
		ybar=.5pt,
		samples=2,
		domain=1:2,
		ylabel=Runtime (s),
		xlabel=Value of $n$,
		ylabel near ticks,
		x label style={ below=3mm},
		xtick=data,
		bar width=5,
		legend pos=north west
		]

		\addplot[ybar,fill=cyan] coordinates {
			(8, 5.060)
			(16, 6.698)
			(32, 7.460)
			(64, 8.877)
			(128, 9.557)
			(256, 10.956)
			
		};
		
		\addplot[ybar,fill=green] coordinates {
			(8, 6.083)
			(16, 7.719)
			(32, 8.527)
			(64, 10.311)
			(128, 11.049)
			(256, 12.459)
			
		};

		\addplot[ybar,fill=yellow] coordinates {
			(8, 7.376)
			(16, 8.516)
			(32, 9.100)
			(64, 11.295)
			(128, 14.841)
			(256, 19.435)
			
		};
		
		\addplot[ybar,fill=red] coordinates {
			(8, 7.673)
			(16, 9.154)
			(32, 10.229)
			(64, 12.975)
			(128, 16.635)
			(256, 22.603)
			
		};
		
		\addplot[ybar,fill=darkgray] coordinates {
			(8, 7.408)
			(16, 8.768)
			(32, 9.863)
			(64, 12.642)
			(128, 15.156)
			(256, 20.447)
			
		};
		\legend{KK13, This paper, IKNP, ALSZ15, NNOB}
		\end{axis}
		\end{tikzpicture}
		\label{fig:runwanallvaryn}
	\end{subfigure}%
	\begin{subfigure}{.34\textwidth}
		\caption{\footnotesize Communication}
		\centering
		\begin{tikzpicture}[scale = 0.55]
		\fontsize{12 pt}{\baselineskip}
		\begin{axis}[
		symbolic x coords={8, 16, 32, 64, 128, 256},
		xtick=data,
		ybar=.5pt,
		samples=2,
		domain=1:2,
		ylabel=Communication (MB),
		xlabel=Value of $n$,
		ylabel near ticks,
		x label style={ below=3mm},
		xtick=data,
		bar width=5,
		legend pos=north west
		]

		\addplot[ybar,fill=cyan] coordinates {
			(8, 1.435)
			(16, 1.913)
			(32, 2.391)
			(64, 2.862)
			(128, 3.34)
			(256, 3.818)
			
		};
		
		\addplot[ybar,fill=green] coordinates {
			(8, 1.435)
			(16, 1.913)
			(32, 2.392)
			(64, 2.863)
			(128, 3.341)
			(256, 3.819)
			
		};

		\addplot[ybar,fill=yellow] coordinates {
			(8, 2.504)
			(16, 3.531)
			(32, 4.891)
			(64, 7.012)
			(128, 10.852)
			(256, 18.503)
			
		};
		
		\addplot[ybar,fill=red] coordinates {
			(8, 4.872)
			(16, 6.692)
			(32, 8.811)
			(64, 11.72)
			(128, 16.356)
			(256, 24.797)
			
		};
		
		\addplot[ybar,fill=darkgray] coordinates {
			(8, 6.339)
			(16, 8.648)
			(32, 11.291)
			(64, 14.678)
			(128, 19.802)
			(256, 28.751)
			
		};
		\legend{KK13, This paper, IKNP, ALSZ15, NNOB}[scale = 0.2]
		\end{axis}			
		\end{tikzpicture}			
		\label{fig:commallvaryn}
	\end{subfigure}
	\label{fig:compareallvaryn}
\end{figure*}
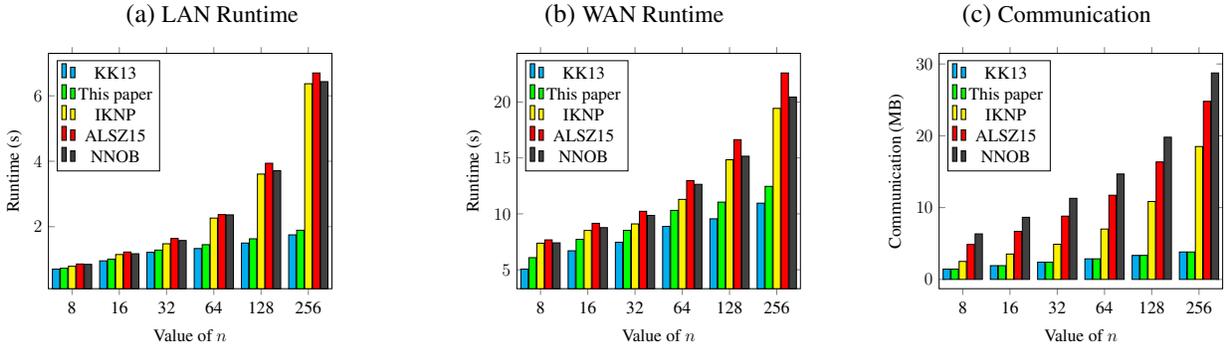

\section{Application to Private Set Intersection}\label{sec:psi}
In a private set intersection~(PSI) protocol, a sender $\sen$ and a receiver $\rec$ hold sets $X = \{x_1, x_2, \ldots x_{n_1}\}$ and $Y= \{y_1, y_2,  \ldots y_{n_2} \}$ respectively.  The goal of the protocol is to let the receiver know the intersection $X \cap Y$ and nothing more. Put simply, a PSI protocol  realizes the functionality $\FPSI(X,Y) = ( \bot, X \cap Y )$. The set sizes are assumed to be public. 

We set our focus on the PSI protocols that are OT-based so that we can employ our OT extension protocol in them to improve efficiency. \cite{PinkasSZ14} introduced  an OT-based PSI protocol  relying on black-box usage of random $\ot{}{}{n}$. Subsequently,  \cite{PinkasSZ15} improved the communication overhead of  \cite{PinkasSZ14}. Both the protocols are semi-honestly secure.  At the core of the protocols lies an important building block called  set inclusion that allows $\rec$ to check whether its input, say $y$, is contained in  $X$, owned by $\sen$, while preserving the input privacy of both the parties.   In the set inclusion protocol, the receiver breaks its $\sigma$-bit element, say $y$  into $t$ blocks of $\eta$-bits. Similarly $\sen$ breaks each of its  $\sigma$-bit element $x_i$ into $t$ blocks of $\eta$-bits. Next, a random  $\ot{}{}{2^\eta}$ is used for $k$th block of receiver's input for $k = 1,\ldots,t$ where the random OT does the following. Denoting $N= 2^\eta$, a random OT of above type generates $N$ random masks and delivers them to $\sen$. $\rec$ receives  from the OT the mask corresponding to its block which acts as its choice string.  $\sen$ then generates a mask for each of its elements in $X$ using the masks received from the $t$ random $\ot{}{}{N}$s. Similarly, $\rec$ combines the masks it receives from the OTs to generate the mask corresponding to its input element $y$.  The verification whether $y$ is included in $X$ is then done by performing checks over the masks. Namely $\sen$ sends across the masks corresponding to all its elements in $X$. $\rec$ verifies if the mask corresponding to $y$ matches with one of them or not.   In a naive approach, PSI can be achieved by having  the receiver run the set inclusion protocol $n_2$ times, once for each element in $Y$. \cite{PinkasSZ14} and subsequently  \cite{PinkasSZ15}  improved  the complexity of the naive approach by reducing the number of OTs  and improving the input length of $\sen$ in the OTs.  Various hashing techniques such as Simple Hashing \cite{PinkasSZ14}  and Cuckoo Hashing  (with a stash $s$ \cite{KirschMW09}),  $h$-ary Cuckoo Hashing   \cite{FotakisPSS03} and Feistel-like functions~\cite{ArbitmanNS10} were used to achieve the goal. 
However, as mentioned before, both \cite{PinkasSZ14} and  \cite{PinkasSZ15} works in semi-honest setting. Indeed, Lamb{\ae}k in his detailed  analysis \cite{cryptoeprint:2016:665} finds   three vulnerabilities  in  \cite{PinkasSZ14,PinkasSZ15} when malicious adversaries are considered. Details follow.


One of vulnerabilities corresponds to sender corruption. Fixing the problem  remains  an open question. The remaining two vulnerabilities correspond to the receiver corruption. In more details, the first problem comes from a malicious receiver who can learn whether some elements of its choice outside his set $Y$ of size $n_2$ belong to $\sen$'s input $X$ or not. The solution proposed in the thesis to thwart this attack uses Shamir secret sharing (SS)  paired with symmetric-key encryption (SKE). Recall that, $\sen$ sends the masks corresponding to its elements in $X$ after the OT executions to help $\rec$ identify the elements in the intersection. The idea of the proposed solution of  \cite{cryptoeprint:2016:665} is  to lock the masks using a key of SKE, secret share the key  and  allow $\rec$ to recover the key only when $\rec$ uses less than or equal to $n_2$ elements in the set inclusion protocols (i.e. in the OT executions).  The second vulnerability may result from any malicious behaviour of $\rec$  in the OT executions of set inclusion protocol.  \cite{cryptoeprint:2016:665} proposes to fix the problem by using maliciously secure (against corrupt receiver) OT protocols. Using off-the-shelves maliciously secure OT extension protocols, \cite{cryptoeprint:2016:665} therefore obtains a PSI protocol that is maliciously secure against corrupt $\rec$ but semi-honestly secure against corrupt  $\sen$.  For complete details of the protocol of Lamb{\ae}k, refer \cite{cryptoeprint:2016:665}. 

We propose to use our maliciously secure OT extension protocol in the PSI protocol of  \cite{cryptoeprint:2016:665} to obtain the most efficient PSI protocol  that is maliciously secure against corrupt $\rec$ but semi-honestly secure against corrupt  $\sen$. As evident from the theoretical and experimental results presented in this work, our maliciously secure OT extension protocol is a better choice compared to the existing maliciously secure extension protocols ~\cite{AsharovL0Z15,NielsenNOB12,KellerOS15}   when the OTs  required are of type $\ot{}{}{n}$. As PSI employs $\ot{}{}{n}$ (instead of  $\ot{}{}{2}$),  our extension protocol fits the bill. Lastly, we find a concrete vulnerability for the malicious corrupt $\rec$  case in  Lamb{\ae}k's PSI protocol when semi-honest KK13 OT extension is used in it. This confirms Lamb{\ae}k's concern of privacy breach of his PSI protocol that may result from privacy breach of the underlying OT protocols and further confirms the necessity of  maliciously secure OT extension in Lamb{\ae}k's PSI protocol. The attack by the corrupt $\rec$ goes as follows: Using the concrete attack discussed in Section~\ref{sec:kk13} for KK13 protocol, a corrupt $\rec$ in the PSI protocol can recover the outputs to $\sen$ in the OT executions. The outputs to $\sen$ are used to compute the masks for the elements of $X$. Therefore by violating the privacy of semi-honest KK13, $\rec$ can completely recover the masks for all the elements of $X$ bypassing the security of  secret sharing technique coupled with SKE.  This allows $\rec$ to learn whether some elements of its choice outside his set $Y$ of size $n_2$ belong to $\sen$'s input $X$ or not. 

\section*{Acknowledgements}
We thank Peter Scholl and Emmanuela Orsini for pointing out a bug in the initial version of the paper. We also thank Michael Zohner and Thomas Schneider for a useful discussion on the implementation part of the work. This work is partially supported by INSPIRE Faculty Fellowship  (DST/INSPIRE/04/2014/015727) from Department of Science \& Technology, India.
{\footnotesize
\bibliographystyle{alpha}
\bibliography{main}
}
\appendix
\section{The Universal Composability (UC) Security Model}\label{sec:UC}
We prove security of our protocol in the standard Universal Composability (UC) framework of Canetti~\cite{Canetti01}, with static corruption. The UC framework introduces a PPT environment $\Env$ that is invoked on the security parameter $\csec$ and an auxiliary input $z$ 
and oversees the execution of a protocol in one of the two worlds.  The ``ideal" world execution involves dummy parties $P_0$ and $P_1$, an ideal adversary $\Sim$ who may corrupt one of the dummy parties, and a  functionality $\Func$. The ``real" world execution involves the PPT parties   $P_0$  and $P_1$  and a real world adversary $\Adv$ who may corrupt 
one of the parties.   The environment $\Env$~chooses the input of the parties and may interact with the ideal/real adversary during the execution. At the   end of the execution, it has to decide upon and output whether a real or an ideal world execution has taken place. 

We let $\Ideal_{\Func, \Sim,\Env}(1^\csec,z)$ denote the random variable describing the output of the environment $\Env$ after interacting with the ideal execution with adversary $\Sim$, the functionality $\Func$, on the security parameter $1^\csec$ and $z$.  Let $\Ideal_{\Func, \Sim,\Env}$ denote the ensemble $\{\Ideal_{\Func, \Sim,\Env}(1^\csec,z)\}_{\csec \in \N,z \in \bitset^*}$. Similarly let $\Real_{\Pi, \Adv,\Env}(1^\csec,z)$ denote the random variable describing  the output of the environment $\Env$~after interacting in a  real execution of a protocol  $\Pi$ with adversary $\Adv$, the parties, on the security parameter $1^\csec$ and $z$.  Let $\Real_{\Pi, \Adv,\Env}$ denote the ensemble 
$\{\Real_{\Pi, \Adv,\Env}(1^\csec,z)\}_{\csec \in \N,z \in \bitset^*}$.

\begin{definition}
	For $n\in \N$, let $\Func$ be a functionality and let $\Pi$ be an $2$-party protocol. We say that  $\Pi$ {\it securely realizes} $\Func$~if for every  PPT real world adversary $\Adv$, there exists a PPT ideal world adversary $\Sim$, corrupting the same parties, such that the following two distributions  are computationally indistinguishable:
	\[\Ideal_{\Func, \Sim, \Env}  \stackrel{c}{\approx}  \Real_{\Pi, \Adv, \Env}. \] 
\end{definition}

\noindent{\bf The $\Func$-hybrid model.}
In order to construct some of our protocols, we will use secure two-party protocols as subprotocols. The standard way of doing this is to work in a ``\emph{hybrid model}'' where both the parties interact with each other (as in the real model) in the outer protocol and use ideal functionality calls (as in the ideal model) for the subprotocols. Specifically, when constructing a protocol $\Pi$ that uses a subprotocol for securely computing some functionality $\Func$,  the parties run $\Pi$ and use ``ideal calls'' to $\Func$ (instead of running the subprotocols implementing $\Func$). 
The execution of $\Pi$ that invokes $\Func$ every time it requires to execute the subprotocol implementing $\Func$ is called   the {\em $\Func$-hybrid execution of $\Pi$} and is denoted as $\Pi^\Func$. The hybrid ensemble $\Hyb_{\Pi^\Func,\Adv,\Env}(1^\csec,z)$ describes $\Env$'s output after interacting with $\Adv$ and the parties $P_0$, $P_1$ running protocol $\Pi^\Func$. By UC definition, the hybrid ensemble should be indistinguishable from the real ensemble with respect to protocol $\Pi$ where the calls to $\Func$ are instantiated with a realization of $\Func$.

\end{document}